\documentclass[11pt,a4paper]{article}
\usepackage[utf8]{inputenc}
\usepackage{amsmath,amsthm}
\usepackage{amsfonts}
\usepackage{amssymb}
\usepackage{graphicx}
\usepackage[margin=1in]{geometry}
\usepackage{algorithmic}
\usepackage{mathtools}
\usepackage[ruled, linesnumbered]{algorithm2e}
\SetKwComment{Comment}{$\triangleright$\ }{}
\SetKwFor{Loop}{Loop}{}{}

\usepackage{hyperref}
\makeatletter
\def\namedlabel#1#2{\begingroup
    #2%
    \def\@currentlabel{#2}%
    \phantomsection\label{#1}\endgroup
}
\makeatother

\newtheorem{theorem}{Theorem}[]
\newtheorem{corollary}{Corollary}[section]
\newtheorem{claim}{Claim}[section]
\newtheorem{definition}{Definition}[section]
\newtheorem{lemma}{Lemma}[section] 

\newtheorem*{remark}{Remark}

\newcommand{\eps}{\varepsilon}
\newcommand{\lis}{\mathsf{LIS}}
\newcommand{\lcs}{\mathsf{LCS}}
\newcommand{\ed}{\mathsf{ED}}
\newcommand{\polylog}{\mathsf{polylog}}

\begin{document}

\begin{titlepage}
\def\thepage{}

\title{Space Efficient Deterministic Approximation of String Measures}

\author{Kuan Cheng \thanks{ckkcdh@hotmail.com \ Computer Science Department, University of Texas at Austin.\ Supported by a Simons Investigator Award (\#409864, David Zuckerman) and NSF Award CCF-1617713.} \and Zhengzhong Jin \thanks{zjin12@jhu.edu\ Department of Computer Science, Johns Hopkins University.\ Supported by NSF Award CCF-1617713 and NSF CAREER Award CCF-1845349.}\and Xin Li  \thanks{lixints@cs.jhu.edu \ Department of Computer Science, Johns Hopkins University.\ Supported by NSF Award CCF-1617713 and NSF CAREER Award CCF-1845349.} \and Yu Zheng \thanks{yuzheng@cs.jhu.edu\ Department of Computer Science, Johns Hopkins University.\ Supported by NSF Award CCF-1617713 and NSF CAREER Award CCF-1845349.}}

\maketitle \thispagestyle{empty}
\begin{abstract}
We study approximation algorithms for the following three string measures that are widely used in practice: edit distance (ED), longest common subsequence (LCS), and longest increasing sequence (LIS).\ All three problems can be solved exactly by standard algorithms that run in polynomial time with roughly $\Theta(n)$ space, where $n$ is the input length, and our goal is to design deterministic approximation algorithms that run in polynomial time with significantly smaller space. 

Towards this, we design several algorithms that achieve $1+\eps$ or $1-\eps$ approximation for all three problems, where $\eps>0$ can be any constant and even slightly sub constant. Our algorithms are flexible and can be adjusted to achieve the following two regimes of parameters: 1) space $n^{\delta}$ for any constant $\delta>0$ with running time essentially the same as or slightly more than the standard algorithms;  and 2) space $\mathsf{polylog}(n)$ with (a larger) polynomial running time, which puts the approximation versions of the three problems in Steve's class (\textbf{SC}).\ Our algorithms significantly improve previous results in terms of space complexity, where all known results need to use space at least $\Omega(\sqrt{n})$. Some of our algorithms can also be adapted to work in the asymmetric streaming model \cite{saks2013space}, and output the corresponding sequence. Furthermore, our results can be used to improve a recent result by Farhadi et. al. \cite{farhadi2020streaming}, which gives an algorithm for ED in the asymmetric streaming model that achieves a $O(2^{1/\delta})$ approximation using $\tilde{O}(n^{\delta}/\delta)$ space, reducing the running time from being \emph{exponential} in  \cite{farhadi2020streaming} to a polynomial.


Our algorithms are based on the idea of using recursion as in Savitch's theorem \cite{Savitch70}, and a careful adaption of previous techniques to make the recursion work. Along the way we also give a new logspace reduction from longest common subsequence to longest increasing sequence, which may be of independent interest.
\end{abstract}
\end{titlepage}

\section{Introduction}
Strings are fundamental objects in computer science, and problems related to strings are among the most well studied problems in the literature. In this paper, we consider the problem of approximating the following three classical string measures:
\begin{description}
\item [Edit distance:] given two strings, the edit distance (ED) between these strings is the minimum number of insertions, deletions, and substitutions to transform one string into another.
\item [Longest common subsequence:] given two strings, the longest common subsequence (LCS) between these strings is the longest subsequence that appears in both strings. 
\item [Longest increasing subsequence:] given one string and a total order over the alphabet, the longest increasing subsequence (LIS) is the longest sequence in the string that is in an increasing order.
\end{description}

These problems have found applications in a wide range of applications, including bioinformatics,  text processing, compilers, data analysis and so on. As a result, all of them have been studied extensively. Specifically, suppose the length of each string is $n$, then both ED and LCS can be computed in time $O(n^2)$ and space $O(n)$ using standard dynamic programming. For LIS, it is known that it can be computed exactly in time $O(n \log n)$ with space $O(n \log n)$. However, in practical applications these problems often arise in situations of huge data sets, where the magnitude of $n$ can be in the order of billions (for example, when one studies human gene sequences). Thus, even a running time of  $\Theta(n^2)$ can be too costly. Similarly, even a $\Theta(n)$ memory consumption can be infeasible in many applications, especially for basic tasks such as ED, LCS, and LIS since they are often used as building blocks of more complicated algorithms.

Motivated by this, there have been many attempts at reducing the time complexity of computing ED and LCS, however none of these attempts succeeded significantly. Recent advances in fine grained complexity provide a justification for these failures, where the work of Backurs and Indyk \cite{BI15} and the work of Abboud, Backurs, and Williams \cite{ABW15} show that no algorithm can compute ED or LCS in time $O(n^{1.99})$ unless the strong  Exponential time hypothesis \cite{IPZ01} is false. Sine then, the focus has been on developing \emph{approximation} algorithms for ED and LCS with significantly better running time, and there has been much success here. In particular, following a recent breakthrough result \cite{CDGKS18}, which gives the first constant factor approximation of ED in truly sub-quadratic time, subsequent improvements have finally achieved a constant factor approximation of ED in near linear time  \cite{brakensiek2019constant, koucky2019constant, AndoniN20}. For LCS the situation appears to be harder, and the best known randomized algorithm \cite{hajiaghayi2019lcs} only achieves an $O(n^{0.498})$ approximation using linear time, which slightly beats the trivial $O(\sqrt{n})$ approximation obtained by sampling.\ Additionally, there is a trivial linear time algorithm that can approximate LCS within a factor of $1/|\Sigma|$ where $\Sigma$ is the alphabet of the strings. A recent work \cite{RSSS19} further provides a randomized algorithm in truly sub-quadratic time that achieves an approximation factor of $O(\lambda^3)$, where $\lambda$ is the ratio of the optimal solution size over the input size. Another recent work by Rubinstein and Song \cite{rubinstein2020reducing} shows how to reduce LCS to ED for binary strings, and uses the reduction to give a near linear time $\frac{1}{2} + \varepsilon$ approximation algorithm for LCS of binary strings, where $\varepsilon>0$ is some constant.

Despite these success, the equally important question of approximating ED and LCS using small space has not been studied in depth. Only a few previous works have touched on this topic, but with a different focus. For example, assume the edit distance between two strings is at most $k$, the work of Chakraborty et. al. \cite{Chakraborty2015LowDE} provides a randomized streaming algorithm that obtains an $O(k)$ approximation of ED, using linear time and $O(\log n)$ space.  Based on this, the work of Belazzougui and Zhang \cite{BelazzouguiZ16} provides a randomized streaming algorithm for computing ED and LCS \emph{exactly} using polynomial time and $\mathsf{poly}(k \log n)$ space. More generally, inspired by the work of  Andoni et. al. \cite{AKO10}, Saks et. al. \cite{saks2013space} studies the asymmetric data streaming model. This model allows one pass streaming access to one string (say $x$), but random access to the other string (say $y$). \cite{saks2013space} gives a $1+\eps$ deterministic approximation of ED in this model using space $O(\sqrt{(n \log n)/\eps})$, although their definition of ED does not allow substitution. So in the standard definition,  \cite{saks2013space} gives a $2+\eps$ deterministic approximation of ED in the asymmetric streaming model using space $O(\sqrt{(n \log n)/\eps})$. Additionally,  \cite{saks2013space} also gives a randomized algorithm that achieves an $\eps n$ additive approximation of LCS in this model, using space $O(k \log^2 n/\eps)$ where $k$ is the maximum number of times any symbol appears in $y$. 

For LIS the situation is slightly better. In particular, the work of Gopalan et. al. \cite{gopalan2007estimating} provides a deterministic streaming algorithm that approximates LIS to within a $1-\eps$ factor, using time $O(n \log n)$ and space $O(\sqrt{n/\eps}\log n)$; while a very recent work by Kiyomi et. al. \cite{kiyomi2018space} obtains a deterministic algorithm that computes LIS \emph{exactly} using $O(n^{1.5} \log n)$ time and $O(\sqrt{n} \log n)$ space. 

In this paper we seek to better understand the space complexity of these problems, while at the same time maintaining a polynomial running time. The first and most natural goal would be to see if we can compute for example ED and LCS exactly using significantly smaller space (i.e., truly sub-linear space of $n^{1-\alpha}$ for some constant $\alpha>0$). However, this again appears to be hard as no success has been achieved in the literature so far. Thus, we turn to a more realistic goal --- to approximate ED and LCS using significantly smaller space. For LIS, our goal is to use approximation to further reduce the space complexity in \cite{gopalan2007estimating} and \cite{kiyomi2018space}.

More broadly, the questions studied in this paper are closely related to the general question of \emph{non-deterministic small space computation} vs. \emph{deterministic small space computation}. Specifically, the decision versions of all three problems (ED, LCS, LIS) can be easily shown to be in the class $\mathsf{NL}$ (i.e., non-deterministic log-space), and the question of whether $\mathsf{NL}=\mathsf{L}$ (i.e., if non-deterministic log-space computation is equivalent to deterministic log-space computation) is a major open question in complexity theory. Note that if $\mathsf{NL}=\mathsf{L}$, this would trivially imply polynomial time algorithms for \emph{exactly} computing ED, LCS, and LIS  in logspace.\ However, although we know that $\mathsf{NL} \subseteq \mathsf{P}$ and  $\mathsf{NL} \subseteq \mathsf{SPACE}(\log^2 n)$ (by Savitch's theorem \cite{Savitch70}), it is not known if every problem in $\mathsf{NL}$ can be solved simultaneously in polynomial time and polylog space, i.e., if $\mathsf{NL} \subseteq \mathsf{SC}$ where $\mathsf{SC}$ is Steve's class. In fact, it is not known if an $\mathsf{NL}$-complete language (e.g., directed $s$-$t$ connectivity) can be solved simultaneously in polynomial time and strongly sub linear space (i.e., space $n^{1-\alpha}$ for some fixed constant $\alpha>0$).\ Thus, studying special problems such as ED, LCS, and LIS, and the relaxed version of approximation is a reasonable first step towards major open problems.

We show that we can indeed achieve our goals. Specifically, for all three problems ED, LCS, and LIS, we give efficient \emph{deterministic} approximation algorithms that can achieve $1+\eps$ or $1-\eps$ approximation, using significantly smaller space than all previous works. In fact, we can even achieve $\polylog(n)$ space while maintaining a polynomial running time. By relaxing the space complexity to $n^{\delta}$ for any constant $\delta>0$, we obtain algorithms whose running time is essentially the same or only slightly more than the standard dynamic programming approach. This is in sharp contrast to the time complexity of ED and LCS, where we only know how to beat the standard dynamic programing significantly by using \emph{randomized} algorithms. 

We have the following theorems.

\begin{theorem}
	\label{thm:edit_distance}
	Given any strings $x$ and $y$ each of length $n$, there are deterministic algorithms that approximate $\ed(x,y)$ with the following parameters:
	\begin{enumerate}
	\item An algorithm that computes a $1+O(\frac{1}{\log\log n})$ approximation of $\ed(x,y)$ with $O(\frac{\log^4 n}{\log \log n})$ bits of space in $O(n^{7+o(1)})$ time. 
	\item For any constants $\delta\in (0,\frac{1}{2})$ and $\eps\in (0,1)$, an algorithm that outputs a $1+\eps$ approximation of $\ed(x,y)$ with $\tilde{O}_{\eps,\delta}(n^\delta)$ bits of space in $\tilde{O}_{\eps,\delta}(n^2)$ time.
\end{enumerate}
The second algorithm can be adapted to work in the asymmetric streaming model with $\tilde{O}_{\eps}(\sqrt{n})$ bits of space.
\end{theorem}

Note that our second algorithm for ED uses roughly the same running time as the standard dynamic programming, but much smaller space. Indeed, we can use space $n^{\delta}$ for any constant $\delta>0$. This also significantly improves the previous result of \cite{saks2013space}, which needs to use space $\Omega(\sqrt{n }\log n)$. Even in the asymmetric streaming model, our result provides a $1+\eps$ approximation for ED instead of $2+\eps$ as in \cite{saks2013space}. With a larger (but still polynomial) running time, we can achieve space complexity $O(\frac{\log^4 n}{\log \log n})$. Next we have the following theorem for LCS.

\begin{theorem}
	\label{thm:lcs}
	Given any strings $x$ and $y$ each of length $n$, there are deterministic algorithms that approximate $\lcs(x,y)$ with the following parameters:
	\begin{enumerate}
	\item An algorithm that computes and outputs a sequence which is a $1-O(\frac{1}{\log \log n})$ approximation of $\lcs(x,y)$, with $O(\frac{\log^4 n}{\log \log n})$ bits of space in $O(n^{6+o(1)})$ time. 
	\item For any constants $\delta\in (0,\frac{1}{2})$ and $\eps\in (0,1)$,an algorithm that computes a $1-\eps$ approximation of $\lcs(x,y)$ with $\tilde{O}_{\eps,\delta}(n^\delta)$ bits of space in $\tilde{O}_{\eps, \delta}(n^{3-\delta})$ time. Furthermore the algorithm can output such a sequence with $\tilde{O}_{\eps, \delta}(n^\delta)$ bits of space in $\tilde{O}_{\eps, \delta}(n^{3}) $ time. 

\end{enumerate}
In addition, the special case of $d=2$ for the second algorithm can be be adapted to work in the asymmetric streaming model, with $O_{\eps}(n^{2.5}\log n)$ time and $O_{\eps}(\sqrt{n}\log n)$ bits of space.
\end{theorem}


To the best of our knowledge, Theorem~\ref{thm:lcs} is the first $1-\eps$ approximation of LCS using truly sub-linear space, and in fact we can achieve space $n^{\delta}$ for any constant $\delta>0$ with only a slightly larger running time than the standard dynamic programming approach. We can achieve space $O(\frac{\log^4 n}{\log \log n})$ with an even larger (but still polynomial) running time. The $(1-\eps)$ approximation in the asymmetric streaming model using space $O_{\eps}(\sqrt{n}\log n)$ is also the first known algorithm to achieve this, and this is incomparable to the result of LCS in \cite{saks2013space}.

For LIS, we also give efficient deterministic approximation algorithms that can achieve $1-\eps$ approximation, with better space complexity than that of \cite{gopalan2007estimating} and \cite{kiyomi2018space}. In particular, we can achieve space $n^{\delta}$ for any constant $\delta>0$ and even space $O(\frac{\log^4 n}{\log \log n})$. We have the following theorem.

\begin{theorem}
	\label{thm:lis}
	Given any string $x$ of length $n$, there are deterministic algorithms that approximate $\lis(x)$ with the following parameters:
	\begin{enumerate}
	\item An algorithm that computes and outputs a sequence which is a $1-O(\frac{1}{\log \log n})$ approximation of $\lis(x)$ with $O(\frac{\log^4 n}{\log \log n})$ bits of space in $O(n^{5+o(1)})$ time.
	\item For any constants $\delta\in (0,\frac{1}{2})$ and $\eps\in (0,1)$, an algorithm that computes a $1-\eps$ approximation of $\lis(x)$ with $\tilde{O}_{\eps,\delta}(n^\delta)$ bits of space in $\tilde{O}_{\eps, \delta}(n^{2-2\delta})$ time. Furthermore the algorithm can output such a sequence with $\tilde{O}_{\eps, \delta}(n^\delta)$ bits of space in $\tilde{O}_{\eps, \delta}(n^{2-\delta}) $ time. 

\end{enumerate}
\end{theorem}

In particular, our theorems directly imply that the approximation versions of the three problems ED, LCS, and LIS are in the class $\mathsf{SC}$:

\begin{corollary}
The problems of achieving a $1+O(\frac{1}{\log \log n})$ or $1-O(\frac{1}{\log \log n})$ approximation  of ED, LCS, and LIS are in the class $\mathsf{SC}$.
\end{corollary}

\begin{remark}
	Our algorithms for LIS also work for the problem of longest non-decreasing subsequence. This is due to the following reduction. Given the original sequence $x \in \Sigma^n$, we change it to a new sequence $y$ where $y_i = (x_i,i)\in \Sigma\times [n] $. We define a new total order on the set $\Sigma\times [n] $ such that $(x_i,i)<(x_j,j)$ if $x_i<x_j$, or, $x_i=x_j$ and $i<j$. Then it is easy to see $\lis(y)$ is equal to the length of the longest non-decreasing subsequence in $x$. This reduction is also a logspace reduction.
\end{remark}

\begin{remark}
In all our theorems, the parameter $\eps$ can actually be slightly sub-constant, i.e., $o(1)$.
\end{remark}

\paragraph{Independent work.} Our results in the asymmetric streaming model are also achieved in a recent independent work by  Farhadi et.\ al.\ \cite{farhadi2020streaming}. Furthermore, \cite{farhadi2020streaming} gives an algorithm for ED in the asymmetric streaming model that achieves a $O(2^{1/\delta})$ approximation using $\tilde{O}(n^{\delta}/\delta)$ space, at the price of using \emph{exponential} running time. However, \cite{farhadi2020streaming} does not give our main results in the non streaming model, where we can achieve $1+\eps$ or $1-\eps$ approximation for all of ED, LCS, LIS using space $n^{\delta}$ or even $\polylog(n)$. Furthermore, we show in Section~\ref{sec:asymmetric} that our results can be used to improve the above approximation algorithm for ED in the asymmetric streaming model in \cite{farhadi2020streaming}, reducing the time from being exponential to a polynomial. Specifically, we have the following theorem:

\begin{theorem}
	Given two strings $x$ and $y$ each of length $n$, where we have streaming access to $x$ and random access to $y$. For any constants $\delta\in (0,\frac{1}{2})$, there is a deterministic algorithm that, making one pass of $x$, outputs a $O(2^{\frac{1}{\delta}})$-approximation of $\ed(x,y)$ in $\tilde{O}_{\delta}(n^{4})$ time with $O(\frac{n^\delta}{\delta}\log n)$ bits of space.
\end{theorem}

\subsection{Technical Overview}
The starting point of all our space efficient approximation algorithms is the well known Savitch's theorem \cite{Savitch70}, which roughly shows that any non-deterministic algorithm running in space $s \geq \log n$ can be turned into a deterministic algorithm running in space $O(s^2)$ by using recursion. Since all three problems of ED, LCS, and LIS can be computed exactly in non-deterministic logspace, this trivially gives deterministic algorithms that compute all of them exactly in space $O(\log^2 n)$. However in the naive way, the running time of these algorithms become quasi-polynomial. 

To reduce the running time, we turn to approximation. Here we use two different sets of ideas. The first set of ideas applies to ED. Note that the reason that the above algorithm for computing ED runs in quasi-polynomial time, is that in each recursion we are computing the ED between \emph{all} possible substrings of the two input strings. To avoid this, we use an idea from \cite{hajiaghayi2019massively}, which shows that to achieve a good approximation, we only need to compute the ED between some carefully chosen substrings of the two input strings. Using this idea in each level of recursion gives us the space efficient approximation algorithms for ED.

The second set of ideas applies to LCS and LIS. Here, we first give a small space reduction from LCS to LIS, and then we can focus on approximating LIS. Again, the reason that the naive $O(\log^2 n)$ space algorithm for LIS runs in quasi-polynomial time, is that in each recursion we are looking at \emph{all} possible cases of breaking the input string into two substrings, computing the LIS in the two substrings which ends and starts at the break point, and taking the maximum of the sums.\ To get an approximation, we use the \emph{patience sorting} algorithm for computing LIS exactly \cite{aldous1999longest}, and the modification in \cite{gopalan2007estimating} which gives an approximation of LIS using smaller space by equivalently looking at only some carefully chosen cases of breaking the input string into two substrings.\ The rough idea is then to use the algorithm in \cite{gopalan2007estimating} recursively, but making this work requires significant modification of the algorithm in \cite{gopalan2007estimating}, both to make the recursion work and to make it work under the reduction from LCS to LIS. 

We now give more details below.

\subsubsection{Edit Distance}

As discussed before, our approximation algorithm for ED is based on recursion. In each level of recursion, we use an idea from \cite{hajiaghayi2019massively} to approximate the edit distance between certain pairs of substrings. We start by giving a brief description of the algorithm in \cite{hajiaghayi2019massively}.

Let $x$ and $y$ be two input strings each of length $n$. Assume we want to get a $(1+\eps)$-approximation of $\ed(x,y)$ where $\eps$ is a parameter in $(0,1)$. Let $b$ be another parameter which we will pick later. The algorithm start by guessing a value $\Delta\leq n$ which is supposed to be a $(1+\eps)$-approximation of $\ed(x,y)$. If this is true, then the algorithm will output a good approximation of $\ed(x,y)$. To get rid of guessing, we can try every $\Delta\leq n$ such that $ \Delta = \lceil (1+\eps)^i\rceil$ for some integer $i$ and take the minimum. This does not affect the space complexity, and only increases the running time by a $\log_{1+\eps}n $ factor.


Given such a $\Delta$, we first divide $x$ into $b$ blocks each of length $\frac{n}{b}$. Denote the $i$-th block of $x$ by $x^i = x_{[l_i,r_i]}$. For simplicity, we fix an optimal alignment between $x$ and $y$ such that $x_{[l_i,r_i]}$ is matched to the substring $y_{[\alpha_i,\beta_i]}$, and the intervals $[\alpha_i,\beta_i]$ are disjoint and span the entire length of $y$. We say that an interval $[\alpha', \beta']$ is an \emph{$(\eps,\Delta)$-approximately optimal candidate} of the block $x^i = x_{[l_i,r_i]}$ if the following two conditions hold:
\[\alpha_i\leq \alpha'\leq \alpha_i+\eps\frac{\Delta}{b}\]
\[\beta_i-\eps\frac{\Delta}{b}-\eps \ed(x_{[l_i,r_i]},y_{[\alpha_i,\beta_i]})\leq \beta'\leq \beta_i\]

\cite{hajiaghayi2019massively} showed that, for each block of $x$ that is not matched to a too large or too small interval in $y$, there is a way to choose $O(\frac{b}{\eps}\log_{1+\eps}n) = O(\frac{b}{\eps^2}\log n)$ intervals such that one of them is an \emph{$(\eps,\Delta)$-approximately optimal candidate}. Then we can compute the edit distance between each block and all of its candidate intervals, which gives $O(\frac{b^2}{\eps^2}\log n)$ values. After this, we can use dynamic programming to find a $(1+O(\eps))$-approximation of the edit distance if $\Delta $ is a $(1+\eps)$-approximation of $\ed(x,y)$. The dynamic programming algorithm takes  $O(\frac{b^3\log n}{\eps^3})$ time. We note that in \cite{hajiaghayi2019massively}, the parameter $b$ is fixed to be $n^{1-\delta}$ for some constant $\delta>0$. However, here we will choose different $b$'s to achieve different regimes of parameters.

Since each block has length $\frac{n}{b}$, computing the edit distance of each block in $x$ with one of its candidate intervals in $y$ takes at most $O(\frac{n}{b\eps}\log n)$ bits of space (we assume each symbol can be stored with space $O(\log n)$). We can run this algorithm sequentially and reuse the space for each computation. Storing the edit distance of each pair takes $O(\frac{b^2}{\eps^2}\log^2 n)$ bits of space. Thus, if we take $b = n^{1/3}$, the above algorithm uses a total of $\tilde{O}_\eps(n^{2/3})$ bits of space. 

We now run the above algorithm recursively to further reduce the space required. Our algorithm takes four inputs: two strings $x,y\in \Sigma^n$, two parameters $b$, $\eps$ such that $b\leq \sqrt n$ and $\eps\in (0,1)$. The goal is to output a good approximation of $\ed(x,y)$ with small space (related to parameters $b$ and $\eps$). Similarly, we first divide $x$ into $b$ blocks. We try every $\Delta$ that is equal to $\lceil (1+\eps)^i\rceil $ for some integer $i$, and for each $\Delta$ there is a set of candidate intervals for each block of $x$. Then, for each block of $x$ and each of its $O(\frac{b}{\eps^2}\log n)$ candidate intervals, instead of computing the edit distance exactly, we recursively call our space efficient approximation algorithm with this pair as input, while keeping $b$ and $\eps$ unchanged.  We argue that if the recursive call outputs a $(1+\gamma)$-approximation of the actual edit distance, the output of the dynamic programming increases by at most a $(1+\gamma)$ factor. Thus if $\Delta$ is a $(1+\eps)$-approximation of $\ed(x,y)$, the output of the dynamic programming is guaranteed to be a $(1+O(\eps))(1+\gamma)$-approximation. The recursion stops whenever one of the input strings has length smaller than $b$. In this case, we compute the edit distance exactly with $O(b\log n)$ space.


Notice that at each level of the recursion, the first input string is divided into $b$ blocks if it has length larger than $b$. Thus the length of first input string at the $i$-th level of recursion is at most $\frac{n}{b^{i-1}}$. The depth of recursion is at most $d = \log_b n$.


At the $d$-th level, our algorithm computes the edit distance exactly. Using this as a base case, we can show that the output of the $i$-th level of recursion is a $(1+O(\eps))^{d-i}$ approximation of the edit distance by induction on $i$ from $d$ to $1$. 

Thus, the output in the first level is guaranteed to be a $(1+O(\eps))^d$-approximation. By $d\leq \log_b n$, the output of our recursive algorithm is a $(1+O(\eps))^d = 1+O(\eps d) = 1+O(\eps \log_b n)$ approximation of $\ed(x,y)$.

For the time and space complexity, we study the recursion tree corresponding to our algorithm. Notice that for each block $x^i$ and each choice of $\Delta$, we consider $O(\frac{b}{\eps^2}\log n)$ candidate intervals. Since there are $b$ blocks and $O(\log_{1+\eps}n)$ choices of $\Delta$, We need to solve $O(\frac{b^2}{\eps^3}\log^3 n)$ subproblems by recursion. Thus, the degree of the recursion tree is $O(\frac{b^2}{\eps^3}\log^3 n)$.

The dynamic programming at each level can be divided into $b$ steps. At the $j$-th step, the inputs are the edit distances between block $x^j$ and each of its candidate intervals. The information we need to maintain is an approximation of edit distances between the first $j-1$ blocks of $x$ and the substrings $y_{[1,l]}$ of $y$, where $l$ is chosen from the set of starting points of the candidate intervals of $x^j$. There are $O(\frac{b}{\eps})$ choices for $l$ and we query the approximated edit distance between $x^j$ and each of its candidate intervals by recursively applying our algorithm. Thus, we only need to maintain $O(\frac{b}{\eps})$ values at any time for the dynamic programming.

At the $i$-th level of recursion, we either invoke one more level of recursion and maintain $O(\frac{b}{\eps})$ values where each value takes $O(\log n)$ bits of space, or do an exact computation of edit distance when one of the input strings has length at most $b$, which takes $O(b\log n)$ bits of space. Hence, the space used at each level is bounded by $O(\frac{b}{\eps}\log n)$. There are at most $d = \log_b n$ levels. The aggregated space used by our recursive algorithm is still $O(\frac{b\log^2 n}{\eps \log b})$. 

We compute the running time by counting the number of nodes in the recursion tree. Notice that the number of nodes at level $i$ is at most $(O(\frac{b^2}{\eps^3}\log^3 n))^{i-1}$. For each leaf node, we do exact computation with time $O(\frac{b^2}{\eps})$, and the number of leaf nodes is bounded by $(O(\frac{b^2}{\eps^3}\log^3 n))^{d-1}$. For each inner node, we run the dynamic programming $O(\log_{1+\eps} n)$ times (the number of choices of $\Delta$) which takes $O(\frac{b^3\log^2 n}{\eps^4})$ time, and the number of inner nodes is bounded by $(d-1)(O(\frac{b^2}{\eps^3}\log^3 n))^{d-2}$.

If we take $b = \log n$ and $\eps = \frac{1}{\log n}$, we get a $(1+O(\frac{1}{\log \log n}))$-approximation algorithm using $O(\frac{\log^4 n}{\log \log n})$ space and $O(n^{7+o(1)})$ time. If we take $b = n^\delta$ for $\delta\in (0,\frac{1}{2})$ and $\eps$ a small constant, we get a $(1+O(\eps))$-approximation algorithm using $\tilde{O}_{\eps, \delta}(n^\delta)$ space and $\tilde{O}_{\eps, \delta}(n^2)$ time.

Our algorithm can also be modified to work in the asymmetric model \cite{saks2013space}. In this model, one has streaming access to one string $x$ and random access to the other string $y$. To see this, notice that the block decomposition of the string $x$ can be viewed as a tree, and for a fixed sequence of $\Delta$ in each level of recursion, the algorithm we discussed above is essentially doing a depth first search on the tree, which implies a streaming computation on $x$. However, the requirement to try all possible $\Delta$ and all candidate intervals may ruin this property since we need to traverse the tree multiple times. To avoid this, our idea is to simultaneously keep track of all possible $\Delta$ and candidate intervals in the depth first search tree on $x$. We stop the recursion and do exact computation whenever each block of $x$ is no larger than $\sqrt{n}$. By doing so, we can still bound the space usage by $\tilde{O}(\sqrt{n})$.

In an independent work \cite{farhadi2020streaming}, the authors give an $O(2^{1/\delta})$ approximation algorithm for edit distance in the asymmetric streaming model with $\tilde{O}(n^\delta /\delta)$ space, at the expense of using exponential running time. We now explain how our algorithm can be used to reduce the running time to a polynomial. 

The algorithm in \cite{farhadi2020streaming} first divides the string $x$ into $b = n^\delta$ blocks. Then, for each block $x^i$, the algorithm recursively finds in $y$ an $\alpha$-approximation of the closest substring to $x^i$, in the following sense: For each block $x^i$, we find a substring $y_{[l_i,r_i]}$ and a value $d_i$, such that for any substring $y_{[l^*, r^*]}$, $\ed(x^i, y_{[l_i,r_i]})\leq d_i \leq \alpha\ed(x^i, y_{[l^*,r^*]}).$ Notice that storing all the $l_i, r_i $ and $d_i$ takes only $O(b\log n)$ space. 

Then, the algorithms tries all possible $1\leq p_0\leq p_1\leq \ldots \leq p_b\leq n+1$ to find the set of $p_i$'s that minimizes $\sum_{i = 1}^{b}\ed(y_{[p_{i-1},p_i)}, y_{[l_i, r_i]})$. Let the optimal set be $\{p^*_i\}$ and record the substring $y_{[l,r]}$ of $y$ where $l=p^*_0$ and $r=p^*_b-1$. Further let $d = \sum_{i = 1}^{b}\ed(y_{[p^*_{i-1},p^*_i)}, y_{[l_i, r_i]})+\sum_{i = 1}^{b}d_i$. \cite{farhadi2020streaming} showed that  $y_{[l,r]}$ together with $d$ is a $(2\alpha+1)$-approximation of the closest substring to $x$ in $y$. Note that this is a recursive algorithm and the recursion stops when the block size of $x$ is at most $n^{\delta}$, at which point we can compute the edit distance exactly using space $\tilde{O}(n^{\delta})$. Since the depth of recursion is $\frac{1}{\delta}$, we are guaranteed to output a $O(2^{\frac{1}{\delta}})$-approximation of the edit distance. 

The super-polynomial (in fact, exponential) running time comes from two parts: First, the step of trying all possible $1\leq p_0\leq p_1\leq \ldots \leq p_b\leq n+1$ to find the set of $p_i$'s that minimizes $\sum_{i = 1}^{b}\ed(y_{[p_{i-1},p_i)}, y_{[l_i, r_i]})$. There are $\binom{n}{n^\delta+1}$ such choices. Second, when computing $\ed(y_{[p_{i-1},p_i)}, y_{[l_i, r_i]})$, they use a $O(\log^2 n)$ space, quasi-polynomial time algorithm guaranteed by Savitch's theorem. 

Our main observation is that, the step of trying all possible $1\leq p_0\leq p_1\leq \ldots \leq p_b\leq n+1$ to find the set of $p_i$'s that minimizes $\sum_{i = 1}^{b}\ed(y_{[p_{i-1},p_i)}, y_{[l_i, r_i]})$, is equivalent to finding the substring of $y$ that minimizes the edit distance to the concatenation of $y_{[l_i,r_i]}$ from $i = 1$ to $b$. Thus, instead of trying all possible $p_i$'s, we can try all substrings of $y$ (there are only $O(n^2)$ such substrings) and compute the edit distance between each substring and the concatenation of the $y_{[l_i,r_i]}$'s. Furthermore, instead of an exact computation which either uses $\Omega(n)$ space or $\log^2 n$ time, we can use our $(1+\eps)$-approximation for ED with $\tilde{O}(n^\delta)$ space and $\tilde{O}(n^2)$ time. The approximation factor is now increased to $O((2+\eps)^{\frac{1}{\delta}})$, which is still $O(2^{\frac{1}{\delta}})$ if we take $\eps = \delta$. But the running time decreases to $\tilde{O}(n^4)$, and the space complexity remains $\tilde{O}(n^\delta/\delta)$.

\subsubsection{Longest Increasing Subsequence}

We now consider the problem of approximating the LIS of a string $x\in \Sigma^n$ over the alphabet $\Sigma$ which has a total order. We assume each symbol in $\Sigma$ can be stored with $O(\log n)$ bits of space. For our discussion, we let $\infty$ and $-\infty$ be two special symbols such that for any symbol $\sigma\in \Sigma$, $-\infty< \sigma< \infty$. We denote the length of the longest increasing subsequence of $x$ by $\lis(x)$. 

Again our algorithm is a recursive one, and in each recursion we use an approach similar to the deterministic streaming algorithm from \cite{gopalan2007estimating} that gives a $1-\eps$ approximation of $\lis(x)$ with $O(\sqrt{n/\eps}\log n)$ space. Before describing their approach, we first give a brief introduction to a classic algorithm for LIS, known as \emph{PatienceSorting}. The algorithm initializes a list $P$ with $n$ elements such that $P[i] = \infty$ for all $i\in [n]$, and then scans the input sequence $x$ from left to right. When reading a new symbol $x_i$, we find the smallest index $l$ such that $P[l] \geq  x_i$ and set $P[l] = x_i$. After processing the string $x$, for each $i$ such that $P[i]< \infty$, we know $\sigma = P[i]$ is the smallest possible character such there is an increasing subsequence in $x$ of length $i$ ending with $\sigma$. We give the pseudocode in algorithm~\ref{patience_sorting_algorithm} and refer readers to \cite{aldous1999longest} for more details about this algorithm.

\begin{algorithm}[H]
	\SetAlgoLined
	\DontPrintSemicolon 
	\KwIn{A string $x \in \Sigma^n$ }
	initialize a list $P$ with $n$ elements such that  $P[i] = \infty$ for all $i\in [n]$ \;
	\For{$i = 1$ \textbf{to} $n$} {
		let $l$ be the smallest index such that $P[l] \geq  x_i$ \;
		$P[l] \gets x_i$ \;
	}
	let $l$ be the largest index such that $P[l] < \infty$\;
	\Return{$l$}\;
	\caption{\emph{PatienceSorting}}
	\label{patience_sorting_algorithm}
\end{algorithm}

We have the following result.

\begin{lemma}
	\label{patience_sorting_lemma}
	Given a string $x$ of length $n$, \emph{PatienceSorting} computes $\lis(x)$ in $O(n\log n)$ time with $O(l\log n)$ bits of space where $l = \lis(x)$.
\end{lemma}

In the streaming algorithm from \cite{gopalan2007estimating}, we maintain a set $S$ and a list $Q$, such that, $Q[i]$ is stored only for $i\in S$ and $S\subseteq [n]$ is a set of size $O(\sqrt{n})$. We can use $S$ and $Q$ as an approximation to the list $P$ in \emph{PatienceSorting} in the sense that for each $s\in S$, there is an increasing subsequence in $x$ of length $s$ ending with $Q[s]$. More specifically, we can generate a list $P'$ from $S$ and $Q$ such that $P'[i] = Q[j]$ for the smallest $j\geq i$ that lies in $S$. For $i$ larger than the maximum element in $S$, we set $P'[i] = \infty$. Each time we read a new element from the data stream, we update $Q$ and $S$ accordingly. The update is equivalent to doing \emph{PatienceSorting} on the list $P'$. When $S$ gets larger than $2\sqrt{n}$, we do a cleanup to $S$ by only keeping $\sqrt{n}/\eps$ evenly picked values from 1 to $\max S$ and storing $Q[s] $ for $s\in S$. Each time we do a ``cleanup", we lose at most $\frac{\eps}{\sqrt{n}}\lis(x)$ in the length of the longest increasing subsequence detected. Since we only do $O(\sqrt{n})$ cleanups, we are guaranteed an increasing subsequence of length at least $(1-\eps)\lis(x)$.

We now modify the above algorithm into another form. This time we first divide $x$ evenly into many small blocks. Meanwhile, we also maintain a set $S$ and a list $Q$. We now process $x$ from left to right, and update $S$ and $Q$ each time we have processed one block of $x$. If the number of blocks in $x$ is small, we can get the same approximation as in \cite{gopalan2007estimating} with $S$ and $Q$ having smaller size. For example, we can divide $x$ into $n^{1/3}$ blocks each of size $n^{2/3}$, and we update $S$ and $Q$ once after processing each block. If we do exact computation within each block, we only need to maintain the set $S$ and the list $Q$ of size $O(\frac{n^{1/3}}{\eps})$. We can still get a $(1-\eps)$ approximation, because we do $n^{1/3}$ cleanups and for each cleanup, we lose about $\frac{\eps}{n^{1/3}}\lis(x)$ in the length of the longest increasing subsequence detected. 

This almost already gives us an $\tilde{O}_{\eps}(n^{1/3})$ space algorithm, except the exact computation within each block needs $O(n^{2/3}\log n)$ space.\ A natural idea to reduce the space complexity is to replace the exact computation with an approximation. When each block $x^i$ has size $n^{2/3}$, running the approximation algorithm from \cite{gopalan2007estimating} takes $O(\frac{n^{1/3}}{\eps}\log n)$ space and thus we can hope to reduce the total space required to $O(\frac{n^{1/3}}{\eps}\log n)$. However, a problem with this is that by running the approximation algorithm on each block $x^i$, we only get an approximation of $\lis(x^i)$. This alone does not give us enough information on how to update $S$ and $Q$. Also, for a longest increasing subsequence of $x$, say $\tau$, the subsequence of $\tau$ that lies in the block $x^i$ may be much shorter than $\lis(x^i)$. This subsequence of $\tau$ may be ignored if we run the approximation algorithm instead of using exact computation. 

We now give some intuition of our approach to fix these issues. Let us consider a longest increasing subsequence $\tau$ of $x$ such that $\tau$ can be divided into many parts, where the $i$-th part $\tau^i$ lies in $x^i$. We denote the length of $\tau^i$ by $d_i$. Let the first symbol of $\tau^i$ be $\alpha_i$ and the last symbol be $ \beta_i$ if $\tau^i$ is not empty. When we process the block $x^i$, we want to make sure that our  algorithm can detect an increasing subsequence of length very close to $d_i$ in $x^i$, where the first symbol is at least $\alpha_i$ and the last symbol is at most $\beta_i$.  We can achieve this by running a  bounded version of the approximation algorithm which only considers increasing subsequences no longer than $d_i$. Since we do not know $\alpha_i$ or $d_i$ in advance, we can guess $\alpha_i$ by trying every symbol in $Q[s]$ where one of them is close enough to $\alpha_i$. For $d_i$, we can try  $O(\log_{1+\eps}n)$ different values of $l$ such that one of them is close enough to $d_i$. In this way, we are guaranteed to detect a good approximation of $\tau_i$.


Based on the above intuition, we now introduce our space-efficient algorithm for LIS called $\textbf{ApproxLIS}$. It takes three inputs, a string $x\in \Sigma^*$, two parameters $b$ and $\eps$. 

We also introduce a slightly modified version of $\textbf{ApproxLIS}$ called $\textbf{ApproxLISBound}$. It takes an additional input $l$, which is an integer at most $n$. We want to guarantee that if the string $x$ has an increasing subsequence of length $l$ ending with $\alpha\in \Sigma$, then $\textbf{ApproxLISBound}(x, b,\eps,l)$ can detect an increasing subsequence of length close to $l$ ending with some symbol no larger than $\alpha$. The idea is to run $\textbf{ApproxLIS}$ but only consider increasing subsequence of length at most $l$. $\textbf{ApproxLISBound}$ has the same space and time complexity as $\textbf{ApproxLIS}$.





We now describe $\textbf{ApproxLIS}$. When the input string $x$ has length at most $b^2$, we compute an $(1-\eps)$-approximation of $\lis$ using the algorithm in \cite{gopalan2007estimating} with $O(\frac{b}{\eps}\log n)$ space. Otherwise, we divide the input string into $b$ blocks each of length $\frac{n}{b}$. Similar to the streaming algorithm in \cite{gopalan2007estimating}, we maintain two sets $S$ and $Q$ of size $O(\frac{b}{\eps})$ as an approximation of the list $P$ when running \emph{PatienceSorting}. We will show that it is enough to use $O(\frac{b}{\eps}\log n)$ bits for $S$ and $Q$, because we only update them $b$ times and we lose about $O(\frac{\eps}{b})\lis(x)$ after each update. Initially, $S$ contains only one element $0$ and $Q[0] = -\infty$. We update $S$ and $Q$ after processing each block of $x$ as follows. 


For simplicity, we denote $S$ and $Q$ after processing the $t$-th block by $ S_t$ and $Q_t$. To see how $S$ and $Q$ are updated, we take the $t$-th update as an example. Given $S_{t-1}$ and $Q_{t-1}$, we first determine the length of the LIS in $x^1\circ \cdots \circ x^t$ that can be detected based on $S_{t-1}$ and $Q_{t-1}$. We denote this length by $k_t$. Notice that, for each $s\in S_{t-1}$, we know there is an increasing subsequence in $x^1\circ \cdots \circ x^{t-1}$ of length $s$ ending with $Q_{t-1}[s]\in \Sigma$. This gives us $|S_{t-1}|$  increasing subsequences. The idea is to find the best extension of  these increasing subsequences in the block $x^t$ and see which one gives us the longest increasing subsequence of $x^1\circ \cdots \circ x^t$. Since each block is of size $\frac{n}{b}$, we cannot always afford to do exact computation. Thus we compute an approximation of the LIS by recursively calling $\textbf{ApproxLIS}$ itself. For each $s\in S_{t-1}$, we run $\textbf{ApproxLIS}(z^s,b,\eps)$ where $z^s$ is the subsequence of $x^t$ with only symbols larger than $Q_{t-1}[s]$. Finally, we let $k_t = \max_{s\in S_{t-1}} (s+\textbf{ApproxLIS}(z^s,b,\eps))$. Given $k_t$, we then set $S_t$ to be the $ \frac{b}{\eps}$-th evenly picked integers from $0$ to $k_t$. 

The next step is to compute $Q_t$. We first set $Q_t[s] = \infty$ for all $s\in S_t$ except $s=0$, and we set $Q_t[0] = -\infty$. Then, for each $s\in S_{t-1}$ and $l = 1,1+\eps,(1+\eps)^2,\ldots, k_t-s$, we run $\textbf{ApproxLISBound}(z^s, b, \eps,l)$. For each $s'\in S_t$ such that $s\leq s'\leq s+l$, we update $Q_t[s']$ if $\textbf{ApproxLISBound}(z^s,b, \eps,l)$ detects an increasing subsequence of length at least $s'-s$ ending with a symbol smaller than the old $Q_t[s']$. The intuition is that, with the bound $l$, we may be able to find a smaller symbol in $\Sigma$ such that there is an increasing subsequence of length $l$ ending with it. This information can be easily ignored if $l$ is a lot smaller than the actual length of LIS in $x^t$. To see why this is important, let $\tau$ be a longest increasing subsequence of $x$, and let $\tau^t$ be the part of $\tau$ that lies in the block $x^t$. The length of $\tau^t$ may be much smaller than the length of LIS in $x^t$. When the bound $l$ is close to $|\tau^t|$, we will be able to detect a good approximation of $\tau^t$ by running $\textbf{ApproxLISBound}(z^s, b, \eps,l)$ on $z^s$ for each $s\in S_{t-1}$. Since we do not know the length of $\tau^t$, we will guess it by trying $O_\eps(\log n)$ values of $l$ and always record the optimal $Q_t[s]$ for $s\in S_t$.

Continue doing this, we get $S_b$ and $Q_b$.\ $\textbf{ApproxLIS}$ outputs the largest element in $S_b$.


$\textbf{ApproxLIS}$ is recursive. We denote the depth of recursion by $d$, and it can be seen that $d$ is at most $\log_b n -1$. To see the correctness of our algorithm, given fixed $b, \eps$, we show the output at the $r$-th recursive level is a $(1-O((d-r)\eps))$-approximation. Thus, the final output will be a $(1-O(\eps\log_b n))$-approximation to $\lis(x)$.

The proof is by induction on $r$ from $d$ to 1. For the base case of $r = d$, the statement follows from \cite{gopalan2007estimating}. Now consider the computation at the $r$-th level. For convenience, we denote the input string by $x$, which has length at most $\frac{n}{b^{r-1}}$. Let $\tau$ be an increasing subsequence of $x$ where $\tau^i$ lies in $x^i$. For our analysis, let $P'_t$ be the list generated by $S_t$ and $Q_t$ in the following way: for every $i$ let $P'_t[i] = Q_t[j]$ for the smallest $j\geq i$ that lies in $S_t$. If no such $j$ exists, set $P'_t[i] = \infty$. Correspondingly, $P_t$ is the list $P$ after running \emph{PatienceSorting} with input $x^1\circ x^2 \circ \cdots \circ x^t$.

Let $h_t = \sum_{j=1}^{t} |\tau^j|$ and $k_t = \max S_t$,  our main observation is the following inequality:
\begin{equation}
\label{tech1}
P'_t[(1-3(d-(r+1)\eps-\eps)h_t-2t\frac{\eps}{b}k_t]\leq P_t[h_t]
\end{equation}
Note that $h_b =|\tau| = \lis(x)$. We have $P_b[h_b]<\infty$ by the correctness of \emph{PatienceSorting}. If inequality~\ref{tech1} holds, then by $k_t\leq h_t$, there must exist an element in $ S_b$ larger than $(1-3(d-r)\eps)\lis(x)$ which directly gives the correctness of the computation at the $r$-th level. 

We prove inequality~\ref{tech1} by induction on $t$. The intuition is that, at the $t$-th update, by the fact that inequality~\ref{tech1} holds for $t-1$, we know that there must exist an $s\in S_{t-1}$ that is close to $h_{t-1}$ and $Q_{t-1}[s]\leq P_{t-1}[h_{t-1}] =\beta_{t-1}<\alpha_t$. By trying $l = 1,1+\eps,(1+\eps),\ldots, k_t-s$, one $l$ is close enough to $|\tau^t|$. Thus we are guaranteed to detect a good approximation of $\tau_t$ in $x^t$ and the inequality also holds for $t$.

For the space complexity, at each recursive level, we need to maintain the sets $S$ and $Q$ with space $O(\frac{b}{\eps}\log n)$. Thus the total space is bounded by $O(d\frac{b}{\eps}\log n) = O(\frac{b\log^2 n }{\eps\log b})$. The analysis of time complexity is similar to the case of edit distance, where we analyze the recursion tree and bound the number of nodes. 

If we take $b = \log n$ and $\eps = \frac{1}{\log n}$, we get a $(1-O(\frac{1}{\log \log n}))$-approximation algorithm using $O(\frac{\log^4 n}{\log \log n})$ space and $O(n^{5+o(1)})$ time. If we take $b = n^\delta$ for $\delta\in (0,\frac{1}{2})$ and $\eps$ a small constant, we get a $(1+O(\eps))$-approximation algorithm using $\tilde{O}_{\eps, \delta}(n^\delta)$ space and $\tilde{O}_{\eps, \delta}(n^{2-2\delta})$ time.  


Our algorithm for approximating the length of LIS can be modified to output a corresponding increasing subsequence, and we call the algorithm $\textbf{LISSequence}$. Roughly, If the input string has length no larger than $b$, we output the LIS exactly using $O(b\log n)$ space. Otherwise, we run our algorithm recursively as follows. Let $\rho $ be the longest increasing subsequence detected by  $\textbf{ApproxLIS}(x,b, \eps)$, thus $\rho$ has length $(1-O(\eps \log_b n))\lis(x)$. We divide $\rho$ into $b$ parts such that the $i$-th part $\rho^i$ lies in $x^i$, thus $\rho^i$ has length at most $|x^i| = \frac{n}{b}$. The idea is to run $\textbf{ApproxLIS}$ and $\textbf{ApproxLISBound}$ multiple times to recover a list $B$ of $b+1$ elements such that the first element of $\rho^i$ is strictly larger than $B[i-1]$ and the last element of $\rho^i$ is at most $B[i]$. Given such a list $B$, we can run $\textbf{LISSequence}$ on the input sequence $x^i$ but only considering elements larger than $B[i-1]$ and at most $B[i]$.

To compute the list $B$, we first set $B[b] $ to be $ Q_b[s_b]$ where $s_b$ is the largest element in $S_b$. This is because $s_b$ is the length of $\rho$ and $Q_b[s_b]$ is the last symbol of $\rho$. Then, we compute the list $B$ from right to left. Once we know $B[i] = Q_i[s_i]$ for some $s_i\in S_i$, we compute $S_{i-1}$ and $Q_{i-1}$ by running $\textbf{ApproxLIS}(x,b,\eps)$ again. Now, for each $s\in S_{i-1}$ and $s\leq s_i$, if we can find an increasing subsequence in $x^{i}$ of length $s_i-s$ with the first symbol larger than $Q_{i-1}[s]$ and the last symbol at most $B[i]$ with algorithm $\textbf{ApproxLISBound}$, we can set $s_{i-1} = s$, $B[i-1]  = Q_{i-1}[s_{i-1}]$ and continue. 

We show that doing this essentially needs us to run $\textbf{ApproxLIS}$ $O(b)$ times. The running time is thus increased by a factor of $O(b)$ compared to the running time of $\textbf{ApproxLIS}$, while the space complexity remains roughly the same.


\subsubsection{Longest Common Subsequence}

Our algorithm for LCS is based on a reduction to LIS. We assume the inputs are two strings $x ,y \in \Sigma^n$. Our goal is to output a $(1-\eps)$-approximation of the LCS of $x$ and $y$.

The reduction is as follows. Given the strings $x$ and $y$, for each $i\in [n]$ let $b^i\in [m]^*$ be the sequence consisting of all distinct indices $j$ in $[m]$ such that $x_i = y_j$, arranged in descending order. Note that $b^i$ may be empty. Let $z = b^1\circ b^2 \circ \cdots \circ b^n$, which has length $O(mn)$ since each sequence $b^i$ is of length at most $m$. We claim that $\lis(z) = \lcs(x,y)$. This is because for every increasing subsequence of $z$, say $t = t_1t_2\cdots t_d$, the corresponding subsequence $y_{t_1}y_{t_2}\cdots y_{t_d}$ of $y$ also appears in $x$. Conversely, for every common subsequence of $x$ and $y$, we can find an increasing subsequence in $z$ with the same length. We call this procedure $\text{ReduceLCStoLIS}$. Note that in our algorithms, $z$ need not be stored, since we can compute each element in $z$ as necessary in logspace by querying $x$ and $y$. Thus our reduction is a logspace reduction.

Once we reduce the LCS problem to an LIS problem, we can use similar techniques as we used for LIS. For our analysis, let $z = \textbf{ReduceLCStoLIS}(x,y)$.

We call our space efficient algorithm for LCS $\textbf{ApproxLCS}$. If one of the input strings is shorter than the parameter $b$, we know $\lcs(x,y)\leq b$. Thus, we can compute $\lis(z)$ using \emph{PatienceSorting} with $O(b\log n)$ bits of space.



Otherwise, the goal is to compute an approximation of $\lis(z)$. One difference compared to $\textbf{ApproxLIS}$ is, instead of dividing $z$ evenly into $b$ blocks, we divide $z$ according to $x$. That is, we first divide $x$ evenly into $b$ blocks. Then $z$ is divided naturally into $b$ blocks such that $z^i = \text{ReduceLCStoLIS}(x^i, y)$. This gives us a slight improvement on running time over the naive approach of running our LIS algorithm after the reduction. Note that $\lis(z^i)$ is at most $\frac{n}{b}$ since the length of $x^i$ is $\frac{n}{b}$. We compute an approximation of $\lis(z^i)$ by recursively calling $\textbf{ApproxLCS}$ with inputs $x,y, b, \eps$. The input size to the next recursive level is decreased by a factor of $b$. The remaining analysis is similar to the case of LIS, and again we can modify our algorithms to output the corresponding common subsequence.

We note our algorithm for LCS can be adapted to work in the asymmetric model with streaming access to $x$ and random access to $y$. We set $b = n^{\frac{1}{2}}$, and the depth of recursion is 2. With random access to $y$, we only need to query $x^i$ to know  $z^i$. Thus, when processing $z^i$, we can keep the corresponding $n^{\frac{1}{2}}$ symbols of $x^i$ in the memory. Since we only process $z^i$ from $i = 1$ to $b$ once, we only need to read $x$ from left to right once. Thus, our algorithm is a streaming algorithm that queries $x$ in one pass with $O(\frac{\sqrt{n}}{\eps}\log n)$ bits of space.

\paragraph{Organization of the paper.}
The rest of the paper is organized as follows. In section~\ref{prelim}, we introduce some notation and give a formal description of the problems studied. In Section~\ref{ed}, we present our algorithms for edit distance. In Section~\ref{lis}, we present our algorithms for LIS. Then in Section~\ref{lcs}, we present our algorithms for LCS. Finally in section~\ref{sec:asymmetric}, we present our algorithms for the asymmetric streaming model. Finally in section~\ref{discussion}, we conclude with discussion on our results and open problems.

\section{Preliminaries}
\label{prelim}

We use the following conventional notations. Let $x \in \Sigma^n$ be a string of length $n$ over alphabet $\Sigma$. By $|x|$, we mean the length of $x$. We denote the $i$-th character of $x$ by $x_i$ and the substring from the $i$-th character to the $j$-th character by $x_{[i,j]}$. We denote the concatenation of two strings $x$ and $y$ by $x\circ y$. By $[n]$, we mean the set of positive integers no larger than $n$. 

\textbf{Edit Distance} The \emph{edit distance} (or \emph{Levenshtein distance}) between two strings $x,y\in \Sigma^*$ , denoted by $\ed (x,y)$, is the smallest number of edit operations (insertion, deletion, and substitution) needed to transform one into another. The insertion (deletion) operation adds (removes) a character at some position. The substitution operation replace a character with another character from the alphabet set $\Sigma$.

\textbf{Longest Common Subsequence} We say the string $s\in \Sigma^t$ is a \emph{subsequence} of $x\in \Sigma^n$ if there exists indices $1\leq i_1<i_2<\cdots < i_t\leq n$ such that $s = x_{i_1}x_{i_2}\cdots x_{i_t}$.  A string $s$ is called a \emph{common subsequence} of strings $x$ and $y$ if $s$ is a subsequence of both $x$ and $y$. Given two strings $x$ and $y$, we denote the length of the longest common subsequence (LCS) of $x$ and $y$ by $\lcs(x,y)$.

\textbf{Longest Increasing Subsequence} In the longest increasing subsequence problem, we assume there is a given total order on the alphabet set $\Sigma$. We say the string $s\in \Sigma^t$ is an \emph{increasing subsequence} of $x\in \Sigma^n$ if there exists indices $1\leq i_1<i_2<\cdots < i_t\leq n$ such that $s = x_{i_1}x_{i_2}\cdots x_{i_t}$ and $x_{i_1} < x_{i_2} < \cdots < x_{i_t}$. We denote the length of the longest increasing (LIS) subsequence of string $x$ by $\lis(x)$. In our analysis, for a string $x$ of length $n$, we assume each element in the string can be stored with space $O(\log n)$. For analysis, we introduce two special symbols $\infty$ and $-\infty$ with $\infty> i$ and $-\infty< i$ for any character $i\in \Sigma$. In our discussion, we let $\infty$ and $-\infty$ to be two imaginary characters such that $-\infty<\alpha<\infty$ for all $\alpha\in \Sigma$.

\section{Edit Distance}

\label{ed}

We now present our algorithms for edit distance. 

\subsection{Some Tools from \cite{hajiaghayi2019massively}}

\label{massively_parallel}

We first introduce some definitions and tools from \cite{hajiaghayi2019massively}, which we will use in our algorithms. To make our paper complete and self-contained, we also include the proofs in Appendix~\ref{proofs}. 

Let $x\in \Sigma^n$ and $y\in \Sigma^m$ be two strings over alphabet $\Sigma$. We assume each symbol of $\Sigma$ can be stored with $O(\log n)$ bits. Our goal is to output a $1+\eps$ approximation for $\ed(x,y)$. Here, $\eps$ is a parameter that can be subconstant. In our algorithm, we only consider the case when $n = \Theta(m)$ since otherwise output a good approximation of $\ed(x,y)$ would be easy. 

We assume an integer $\Delta$ is given to us which is supposed to be a $1+\eps$ approximation of the actual edit distance. Such an assumption only increase the total amount of computation by a $O(\log_{1+\eps}(n))$ factor. This is because we can try all $\Delta = (1+\eps)^i$ with $i\in [\rceil \log_{1+\eps}(n)\rceil] $ and make sure one of $\Delta$ is a $1+\eps$ approximation of $\ed(x,y)$.

Given such a $\Delta$, we first divide string $x$ into $b$ blocks ($b$ is a parameter we will pick later). We denote the $i$-th block by $x^i$. Then for each block, say $x^i$, we determine a set of substrings of $y$ as candidate intervals such that one of the candidate substrings is close to the optimal substring $x^i$ is matched to in the optimal matching. We call such a substring \emph{approximately optimal candidate}. Then, if we know the edit distance (or a good approximation of edit distance) between each block and each of its candidate substring of $y$, we can run a dynamic programming to get a approximation of the actual edit distance distance.


In the following analysis, we fix an optimal alignment between $x$ and $y$ where the $i$-th block $x_{[l_i,r_i]}$ is matched to substring $y_{[\alpha_i,\beta_i]}$ such that the intervals $[\alpha_i,\beta_i]$'s are disjoint and span the entire length of $y$. By the assumption, we have $\ed(x,y) = \sum_{i = 1}^{N}\ed(x^i,y_{[\alpha_i,\beta_i]})$.

We now give the definition of \emph{$(\eps,\Delta)$-approximately optimal candidate}.

\begin{definition}
	[\cite{hajiaghayi2019massively}]
	\label{approx_optimal_candidate}
	We say an interval $[\alpha',\beta']$ is an \emph{$(\eps,\Delta)$-approximately optimal candidate} of the block $x^i = x_{[l_i,r_i]}$ if the following two conditions hold:
	\[\alpha_i\leq \alpha'\leq \alpha_i+\eps\frac{\Delta}{b}\]
	\[\beta_i-\eps\frac{\Delta}{b}-\eps \ed(x_{[l_i,r_i]},y_{[\alpha_i,\beta_i]})\leq \beta'\leq \beta_i\]
\end{definition}

We first show that if $\Delta$ is a good approximation of $\ed(x,y)$, and for each block that is not matched to a too large or a too small interval, we know the edit distance between it and one of its approximately optimal candidate, we can get a good approximation of $\ed(x,y)$. We put it formally in Lemma~\ref{lem:approx_optimal_candidate}.

\begin{lemma}[Implicit from \cite{hajiaghayi2019massively}]
	\label{lem:approx_optimal_candidate}
	Let $\eps\in (0,1)$ ($\eps $ can be subconstant) and $\eps' = \eps/10$. Assume $\ed(x,y)\leq \Delta\leq (1+\eps')\ed(x,y)$.  For each $i\in [b]$, let $(\alpha_i',\beta_i')$ be any $(\eps',\delta)$-approximately optimal candidate of $x^i$. If $\eps'|\alpha_i-\beta_i+1|\leq |x^i|\leq 1/\eps'|\alpha_i-\beta_i +1|$, let $D'_i =|\alpha_i-\alpha'_i|+ \ed(x^i,y_{[\alpha'_i,\beta'_i]}) + |\beta_i-\beta'_i| $. Otherwise, let $D'_i  = |x^i| + |\alpha_i-\beta_i+1|$. Then
	\[\ed(x,y)\leq \sum_{i = 1}^{b}D'_i \leq (1+\eps)\ed(x,y).\]
\end{lemma}

To make our work self-contained, we provide a proof in Appendix~\ref{proofs}.

We now show that, for each $i$ and $\eps, \Delta$, without knowing the optimal alignment, we can pick a small set of candidate intervals such that one of the intervals is an \emph{$(\eps , \Delta)$-approximately optimal candidate} for $x^i$.  

That is, there exist a set of intervals $C^i_{\eps,\Delta}$ with size $O(b\log n/\eps^2)$ and one of the intervals in $C^i_{\eps,\Delta}$ is an \emph{$(\eps , \Delta)$-approximately optimal candidate} for $x^i$. The set $C^i_{\eps,\Delta}$ can be find with the algorithm \emph{CandidateSet} which is implicit from \cite{hajiaghayi2019massively}. The algorithm takes six inputs : three integers $n$, $m$, and $b$, an interval $(l_i,r_i)$, $\eps\in (0,1)$, and $\Delta\leq n$ and outputs set $C^i_{\eps,\Delta}$. Here, $n$ and $m$ are the lengths of string $x$ and $y$ correpondingly. The pseudocode is given in algorithm~\ref{algo:CandidateSet}. 

\begin{algorithm}
	\SetAlgoLined
	\DontPrintSemicolon
	\KwIn{three integers $n$, $m$, and $b$, an interval $(l_i,r_i)$, $\eps\in (0,1)$, and $\Delta\leq n$} 
	$|x^i| = r_i-l_i+1$ \;
	initialize $C$ to be an empty set \;
	
	\ForEach{$i'\in [l_i-\Delta- \eps\frac{ \Delta}{b},l_i+\Delta+  \eps\frac{ \Delta}{b} ]\cap [m]$ such that $i'$ is a multiple of $\lceil \eps \frac{ \Delta}{b}\rceil $ }{
		\Comment{if $\lceil \eps \frac{ \Delta}{b}\rceil=0$, we loop every $i'$ in $[l_i-\Delta-1,l_i+\Delta+1]\cap [m]$}
		
		\ForEach{$j'=0 $ \textbf{or} $j' =  \lceil (1+\eps)^i \rceil \textbf{ for some integer } i\leq \lceil \log_{1+\eps}(m)\rceil  \}$ }{
			\Comment{pick $O(\log_{1+\eps} n)$ ending points}
			\If{ $|x^i|-j' \geq \eps|x^i| $}{
				add $(i',i'+|x^i|-1-j')$ to $C$ \;
			}
			\If{$|x^i|+j' \leq 1/\eps|x^i| $}{
				add $(i',i'+|x^i|-1+j')$ to $C$ \;
			}
		}	
	}
	\Return{$C$}
	
	\caption{CandidateSet}
	\label{algo:CandidateSet}
\end{algorithm}

\begin{lemma}[Implicit from \cite{hajiaghayi2019massively}]
	\label{lem:CandidateSet}
	If $\eps m\leq n \leq \frac{1}{\eps} m$, then $C^i_{\eps,\Delta} = \text{CandidateSet}(n,m,b,(l_i,r_i),\eps,\Delta)$ is of size $ O(\frac{b\log n}{\eps^2})$. For $x^i = x_{[l_i,r_i]}$, if $\eps|\alpha_i-\beta_i+1|\leq |x^i|\leq 1/\eps|\alpha_i-\beta_i+1|$ and $\Delta\geq \ed(x,y)$, then one of the intervals in $C^i_{\eps,\Delta}$ is an \emph{$(\eps, \Delta)$-approximately optimal candidate} of $x^i$.
\end{lemma}

A proof of Lemma~\ref{lem:CandidateSet} can be found in Appendix~\ref{proofs}.


Given the information of edit distances between $x^i$ and each of intervals in $C^i_{\eps,\Delta}$, we can run a simple dynamic programming algorithm \emph{DPEditDistance} to get an approximation of $\ed(x,y)$.  \emph{DPEditDistance} takes six inputs, $n$, $m$, $b$, $\Delta$, $\eps$, and a two dimensional list $M$ such that $M(i, (\alpha,\beta)) = \ed(x^i, y_{[\alpha, \beta]})$ for each $i$ and $(\alpha,\beta)\in  C^i_{\eps,\Delta}$. The pseudocode is given in algorithm~\ref{algo:DPEditDistance}.

\begin{algorithm}
	\SetAlgoLined
	\DontPrintSemicolon
	\KwIn{three integers $n$, $m$, $b$, $\Delta\leq n$, $\eps\in (0,1)$, and a two dimensional list $M$ such that $M(i, (\alpha,\beta)) = \ed(x^i, y_{[\alpha, \beta]})$ for each $i$ and $(\alpha,\beta)\in  C^i_{\eps,\Delta}$.} 
	let $C^i$ be the set of starting points of intervals in $C^{i}_{\eps,\Delta}$ with no repetition for each $i\in [b]$ \;

	\ForEach{$\alpha\in C^1$ }{
		$A(0,\alpha-1) = \alpha-1$ \Comment{$A$ is a two dimensional array for storing the intermediate results of the dynamic programming }\;
	}
	\For{$i=1$ \textbf{to} $b-1$}{
		\ForEach{ $\alpha\in C^{i+1}$  }{
			\(A(i,\alpha-1) = \min \begin{dcases}
			\min_{\alpha'\in C^i, \alpha'\leq \alpha}A(i-1,\alpha'-1)+ |\alpha-\alpha'| +|x^i|\\
			\min_{ \substack{(\alpha',\beta')\in C^i_{\eps,\Delta} \\\text{s.t. } \beta' \leq \alpha-1} }A(i-1,\alpha'-1)+M(i,(\alpha',\beta')) + \alpha-1-\beta'\\
			\end{dcases}
			\) \label{algo:DPEditDistance_updating}\;
		}
	}
	
	\(d = \min
	\begin{dcases}
	\min_{\alpha'\in C^{b}}A(b-1,\alpha'-1)+ |m-\alpha'| + |x^b|\\
	\min_{ \substack{(\alpha',\beta')\in C^b_{\eps,\Delta} \\\text{s.t. } \beta' \leq m} }A(b-1,\alpha'-1)+M(b,(\alpha',\beta')) + m-\beta'\\
	\end{dcases}
	\)\label{DPEditDistance:last}\;
	\Return d
	\caption{DPEditDistance}
	\label{algo:DPEditDistance}
\end{algorithm}

\begin{lemma}
	\label{lem:DPEditDistance}
	For any fixed $\eps$, we let $\eps' = \eps/10$. Assume $\ed(x,y)\leq \Delta\leq (1+\eps')\ed(x,y)$ and for every $(\alpha, \beta)\in C^i_{\eps',\Delta}$, $M(i,(\alpha,\beta)) = \ed(x^i,y_{[\alpha,\beta]})$, then $\text{DPEditDistance}(n,m,b, \eps',\Delta,M)$ outputs a $(1+\eps)$-approximation of $\ed(x,y)$ in $O(\frac{b^3\log n}{\eps^3})$ time with $O(\frac{b}{\eps}\log n)$ bits of space. 
	
	Also, in the input, if we replace $M(i,(\alpha,\beta))$ with a $(1+\gamma)$ approximation of $\ed(x^i,y_{[\alpha,\beta]})$, that is \[\ed(x^i,y_{[\alpha,\beta]})\leq M(i,(\alpha,\beta))\leq (1+\gamma)\ed(x^i,y_{[\alpha,\beta]})\]
	then $\text{DPEditDistance}(n,m,b, \eps',\Delta,M)$ outputs a $(1+\eps)(1+\gamma)$-approximation of $\ed(x,y)$.
\end{lemma}

We note that the space complexity of algorithm~\ref{algo:DPEditDistance}, \emph{DPEditDistance}, is optimized. Let $C^i$ be the set of starting points of intervals in $C^{i}_{\eps,\Delta}$ with no repetition. Notice that when updating $A(i,\alpha)$, we only need the information of $A(i-1, \alpha'-1)$ for every $\alpha'\in C^i$. Thus, we can release the space used to store $A(i-2, \alpha''-1)$ for every $\alpha''\in C^{i-1}$. Furthermore, for line~\ref{DPEditDistance:last}, we only need the information of $A(i-1, \alpha-1)$ for every $\alpha\in C^.$. From algorithm~\ref{algo:CandidateSet}, we know that for each $i$, we pick at most $ b/\eps$ points as the starting point of the candidate intervals. The size of  $C^i$ is at most $b/\eps$. Since each element in $A$ is an integer at most $n$, it can be stored with $O(\log n)$ bits of space. Thus, the space required is $O(\frac{b}{\eps}\log n)$.

A full proof of Lemma~\ref{lem:DPEditDistance} can be found in Appendix~\ref{proofs}.

\subsection{Space Efficient Algorithm for Edit Distance}

Our algorithm recursively use the above ideas with carefully picked parameters. 

In the following, $b$ and $\eps$ are two parameters we will set later. We call our space-efficient approximation algorithm for edit distance \emph{SpaceEfficientApproxED} and give the pseudocode in algorithm~\ref{algo:SpaceEfficientApproxED}. 

\begin{algorithm}[H]
	\SetAlgoLined
	\DontPrintSemicolon
	\KwIn{Two strings $x$ and $y$, parameters $b\leq \sqrt{n}$ and $\eps \in (0,1)$} 
	\If{$|x|\leq b$ \label{len}}{
		compute $\ed(x,y)$ exactly \;
		\Return{$\ed(x,y)$} \;
	}
	$ed\gets \infty$\;
	set $n = |x|$ and $m = |y|$ \;
	divide $x$ into $b$ block each of length at most $\lceil n/b \rceil$ such that $x = x^1\circ x^2 \circ \cdots \circ x^{b}$\;
	
	\ForEach{$\Delta= 0 \textbf{ or } \lceil (1+\eps)^j\rceil \textbf{ for some integer } j \text{ and } \Delta\leq \max\{|x|,|y|\}$ \label{algo:SpaceEfficientApproxEDDeltaLoop}}{
		\For{$ i = 1$ \textbf{to} $b$}{
			\ForEach{$(a,b)\in \text{CandidateSet}(n,m,(l_i,r_i),\eps,\Delta)$}{
				$M(i,(a,b))\gets \text{SpaceEfficientApproxED}(x^i, y_{[a,b]},b, \eps)$\label{algo:SpaceEfficientApproxED_recurse}\;
			}
		}
		$ed\gets \min\{ed,\text{DPEditDistance} (n,m,b,\eps, \Delta,M)\}$\;
	}
	\Return{$ed$}
	\caption{SpaceEfficientApproxED}
	\label{algo:SpaceEfficientApproxED}
\end{algorithm}

We have the following result. 

\begin{lemma}
	\label{lem:SpaceEfficientApproxED}
	Given two strings $x, y\in \Sigma^n$, parameters $b\leq \sqrt{n}$ and $\eps\in (0,1)$, $\text{SpaceEfficientApproxED} (x,y, b,\eps)$ outputs a $1+O(\eps \log_b n)$-approximation of $\ed(x,y)$ with $O(\frac{b\log^2 n }{\eps \log b})$ bits of space in $(\frac{\eps^2\log b}{b\log n}+\frac{\eps^2}{\log^2 n}) (O(\frac{b^2\log^2 n}{\eps^3}))^{\log_b n }$ time. 
\end{lemma}

\begin{proof} [Proof of Lemma~\ref{lem:SpaceEfficientApproxED}]

	Algorithm~\ref{algo:SpaceEfficientApproxED} is recursive. We start from level one and every time \emph{SpaceEfficientApproxED} is called, we enter the next level. We say the largest level we will reach is the maximum depth of recursion. In the following, to avoid ambiguity, $x$ and $y$ denote the input strings at the first level where both string has length $n$.
	
	Notice that the length of first input string at $i$-th level is at most $\frac{n}{b^{i-1}}$. The recursion terminates when the length of first input string is no larger than $b$. Thus the maximum depth of recursion is $\log_b n$. We denote the maximum depth of recursion by $d$. 

	We first show the correctness of our algorithm by prove the following claim.

	\begin{claim}
		At the $l$-th level, the output is a $(1+ 10\eps)^{d-l}$ approximation of the edit distance of its input strings.
	\end{claim} 
	\begin{proof}
		We prove this by induction on $l$ from $d$ to $1$. For the base case $ l= d$, we output the exact edit distance. The claim holds trivially. 
		
		Now, we assume the claim holds for the $(l+1)$-th level. At the level $l$, if the input string $x$ has length no larger than $b$, we output the exact edit distance. The claim holds for level $i$. Otherwise, since we tried every $\Delta= \lceil (1+\eps)^j\rceil \text{ for some integer } j \text{ and } \Delta\leq n+m$, one of $\Delta$ satisfies $\ed(x,y)\leq\Delta\leq(1+\eps)\ed(x,y)$. Denote such a $\Delta$ by $\Delta_0$. For $(a,b)\in C^i_{\eps, delta}$ that $M(i,(a,b))$ is a $(1+10\eps)^{d-(l+1)}$ approximation of $\ed(x^i,y_{[a,b]})$ by the inductive hypothesis. By lemma~\ref{lem:DPEditDistance}, $\text{DPEditDistance} (n,m,b,  \eps,\Delta_0, M)$ outputs a $(1+10\eps)(1+10\eps)^{d-(l+1)}= (1+10\eps)^{d-l}$ approxiamtion of $\ed(x,y)$ when $\Delta  = \Delta_0$. This proves the claim.
		
	\end{proof}
	
	By the above claim, for the first level, our algorithm always output a $(1+10\eps)^{d-1}   = 1+O(\eps d) = 1+O(\eps \log_b n)$ approximation of $\ed(x,y)$.  

	We now turn to the space and time complexity. We can consider our recursion structure as a tree. The first level corresponds to the root of the recursion tree. Notice that we need to try $O(\log_{1+\eps}(n)) $ different $\Delta$. For each $\Delta$, we need to query the next level $O(b(\frac{b\log n}{\eps^2})) $ time. This is because there are $b$ blocks and for each block, we choose $O(\frac{b\log n}{\eps^2})$ candidate intervals by lemma~\ref{lem:CandidateSet}. Thus, the recursion tree has degree $O(\log_{1+\eps}(n)\frac{b^2\log n}{\eps^2})  = O( \frac{b^2\log^2 n}{\eps^3})$ with depth $d \leq \log_b n$.
	
	Running the algorithm essentially has the same order as doing a depth first search on the recursion tree. At each level of the recursion, we only need to remember the information in one node. Thus, total space required is equal to the space needed for one inner node times the depth of recursion tree, plus the space needed for one leaf node. 
	
	For the leaf nodes, the first input string is of length at most $b$, we can compute the edit distance exactly with space $O(b\log n)$ bits of space. 
	
	For other nodes, the task is to run a dynamic programing, i.e. algorithm \emph{DPEditDistance} where we need to query the next level of recursion to get the matrix $M$. In algorithm \emph{DPEditDistance}, the input is a matrix $M$ and we need to compute a matrix $A$. For matrices $A$, the rows are indexed by $i $ from $0$ to $ b$ and for the $i$-th row, the columns are indexed by the elements in set $C^i$. $C^i$ is the set of starting points of intervals in $C^i_{\eps,\Delta}$. By the proof of lemma~\ref{lem:CandidateSet}, $C^i$ is of size $O(\frac{b}{\eps})$.
	
	According to the proof of lemma~\ref{lem:DPEditDistance}, we can divide the dynamic programming into $b$ steps and for each step, we update one row of $A$. When computing the $i$-th row of $A$, we only need to query the $i-1$-th row of $A$ and the $i$-th row of $M$. Thus, the space used to remember previous rows of $A$ can be reused. Also, we only query each element in the $i$-th row of $M$ once, so we do not need to remember matrix $M$ and this does not affect the time complexity. Thus, for each inner node of the recursion, we only need space enough for storing two rows of matrix $A$ and each element of $A$ is an integer no larger than $n$. The space for each inner node is bounded by $O(\frac{b}{\eps}\log n)$ bits. 
	
	Thus, the space complexity of our algorithm is bounded by $O(d\frac{b}{\eps}\log n +  b\log n )= O(\frac{b\log^2 n }{\eps \log b})$ bits.

	For time complexity, we denote the time used for computation at the $i$-th level by $T_i$ (excluding the time used for running \emph{SpaceEfficientApproxED} at the $i$-level). The time complexity is bounded by the sum of time spent at each level. Denote the total running time by $T$, we have $T = \sum_{i = 1}^d T_i$.
	
	Once \emph{SpaceEffientApproxED} is called at the $(i-1)$-th level, we enter the $i$-th level. 
	
	Each time we enter the $i$-th level, there are two possible cases. For the first case, the operation at the $i$-th level is calculating the exact edit distance with one of the input string has length at most $b$ and the other string has length $O(b/\eps)$. It takes $O(\frac{b^2}{\eps})$ time. 
	
	Otherwise, we run $\text{DPEditDistance}$ for $O(\log_{1+\eps} n )= O(\frac{\log n}{\eps})$ times. By lemma~\ref{lem:DPEditDistance}, it takes $O(\frac{b^3\log n}{\eps^3})$ time. 
	
	Thus, each time we enter the $i$-th level, the time required at that level is bounded by $O(\frac{b^3\log^2 n}{\eps^4}) $. Since the recursion tree has degree $O( \frac{b^2\log^2 n}{\eps^3})$, we enter the $i$-th level  $(O(\frac{b^2\log^2 n}{\eps^3}))^{i-1 }$ times. Thus, $T_i = \frac{b^3\log^2 n}{\eps^4}(O(\frac{b^2\log^2 n}{\eps^3}))^{i-1 }$. 
	
	For $1\leq i \leq d-1$, $T_i$ is bounded by $\frac{b^3\log^2 n}{\eps^4}(O(\frac{b^2\log^2 n}{\eps^3}))^{i-1 }$. We have
	\begin{equation}
	\label{ed_a}
		\begin{split}
			\sum_{i = 1}^{d-1}T_i &\leq (d-1)T_{d-1} \\
			& \leq d\frac{b^3\log^2 n}{\eps^4}(O(\frac{b^2\log^2 n}{\eps^3}))^{d-2 }\\
			& = \frac{b^3 \log^3 n}{\eps^4\log b}(O(\frac{b^2\log^2 n}{\eps^3}))^{\log_b n -2 }\\
			&  = \frac{\eps^2\log b}{b\log n}(O(\frac{b^2\log^2 n}{\eps^3}))^{\log_b n  }.
		\end{split}
	\end{equation}
	
	Also notice that at the $d$-th level, we always do the exact computation of edit distance, which takes $O(\frac{b^2}{\eps})$ time. Thus

	\begin{equation}
	\label{ed_b}
		\begin{split}
			T_d &= \frac{b^2}{\eps}(O(\frac{b^2\log^2 n}{\eps^3}))^{d-1 } \\
			& = \frac{\eps^2}{\log^2 n}( O(\frac{b^2\log^2 n}{\eps^3}))^{d } \\
			& = \frac{\eps^2}{\log^2 n} (O(\frac{b^2\log^2 n}{\eps^3}))^{\log_b n }.
		\end{split}	
	\end{equation}
	
	Combining \ref{ed_a} and \ref{ed_b}. We know the running time is bounded by
	$$
	(\frac{\eps^2\log b}{b\log n}+\frac{\eps^2}{\log^2 n}) (O(\frac{b^2\log^2 n}{\eps^3}))^{\log_b n }.
	$$

\end{proof}

\begin{theorem}
	\label{lem:ed_polylog}
	Given two strings $x$ and $y$, both of length $n$, there is a deterministic algorithm that outputs a $1+O(\frac{1}{\log\log n})$ approximation of $\ed(x,y)$ with $O(\frac{\log^4 n}{\log \log n})$ bits of space in $O(n^{7+o(1)})$ time. 
\end{theorem}

\begin{proof}[Proof of Theorem~\ref{lem:ed_polylog}]
	Let $b = \log n$ and $\eps = \frac{1}{\log n}$. Then Theorem~\ref{lem:ed_polylog} is a direct result of Lemma~\ref{lem:SpaceEfficientApproxED}.
\end{proof}

\begin{theorem}
	\label{lem:ed_ndelta}
	Given two strings $x$ and $y$, both of length $n$, for any fixed constant $\eps\in (0,1)$, $\delta\in (0,\frac{1}{2})$, there is a deterministic algorithm that outputs a $1+\eps$ approximation of $\ed(x,y)$ with $\tilde{O}_{\eps,\delta}(n^\delta)$ bits of space in $\tilde{O}_{\eps,\delta}(n^2)$ time
	
\end{theorem}

\begin{proof}[Proof of Theorem~\ref{lem:ed_ndelta}]
	Let $b = n^\delta$ and pick $\eps'$ to be a constant sufficiently smaller than $\eps$. We run algorithm \emph{SpaceEfficientApproxED} with inputs $x$, $y$, $b$, and $\eps'$. Then Theorem~\ref{lem:ed_ndelta} is a direct result of Lemma~\ref{lem:SpaceEfficientApproxED}.
\end{proof}

\section{Longest Increasing Subsequence}
\label{lis}

We now present our space-efficient algorithms for LIS. 

\subsection{Space Efficient Algorithm for LIS}

We call our main algorithm $\textbf{ApproxLIS}$ and give the pseudocode in Algorithm~\ref{algo:approxlis}. In addition to algorithm $\textbf{ApproxLIS}$, we introduce a slightly modified version of it called $\textbf{ApproxLISBound}$. $\textbf{ApproxLISBound}$ takes an additional input $l$, which is an integer at most $n$. We want to guarantee that, if the input sequence $x$ has an increasing subsequence of length $l$ ending with $\alpha\in \Sigma$, then $\textbf{ApproxLISBound}(x,b,\eps,l)$ can detect an increasing subsequence with length close to $l$, and ending with some symbol in $\Sigma$ at most $\alpha$.

$\textbf{ApproxLISBound}$ is similar to $\textbf{ApproxLIS}$ with only a few differences. First, at line~\ref{ApproxLIS:newk} of algorithm $\textbf{ApproxLIS}$, we always require $k$ to be at most $l$. That is, we let $k = \min\{l, \max\{k, s+d\}\}$. Second, instead of output $\max S$, we output the whole set $S$ and list $Q$ (The streaming algorithm from \cite{gopalan2007estimating} also maintains set $S$ and list $Q$). We omit the pseudocode for $\textbf{ApproxLISBound}$.

%

\begin{algorithm}
	\SetAlgoLined
	\DontPrintSemicolon 
	\KwIn{A string $x$ , parameters $b$ and $\eps$. }
	
	\If{ $|x| \leq b^2$}{
		compute an $(1-\eps)$-approximation of $\lis(x)$ with the streaming algorithm from \cite{gopalan2007estimating} using $O(\frac{b}{\eps}\log n)$ space\;
		\Return{}
	}

	divide $x$ evenly into $b$ blocks such that $x = x^1\circ x^2\circ \cdots \circ x^b$ \Comment*[f]{$|x^i|\leq \lceil n/b \rceil $}\;
	
	initialize $S =\{0\}$ and $Q[0] = -\infty$ \;
	
	\For{$i = 1$ \textbf{to} $b$}{
		
		$k = 0$\;
		
		\ForEach{$s\in S$}{
			let $z$ be the subsequence of $x^i$ by only considering the elements larger than $Q[s]$\;
			$d = \textbf{ApproxLIS} (z, b, \eps)$\;
			$k = \max\{k,s+d \}$ \label{ApproxLIS:newk}\;
		}
		
		\uIf {$k\leq b/\eps$}{
			 let $S' = \{0,1,2,\ldots,k\}$\;
		 }
	 	\Else{
	 		let $S' = \{0,\frac{\eps}{b}k,2\frac{\eps}{b}k,\ldots ,k\}$ \Comment*[f]{evenly pick $ b/\eps + 1$ integers from $0$ to $k$ (including $0$ and $k$) }
 		
 		} 
		
		$Q'[s] = \infty$ for all $s'\in S'$ except $Q'[0] = -\infty$\;
		\ForEach{$s\in S$}{
			let $z$ be the subsequence of $x^i$ by only considering the elements larger than $Q[s]$\;
			\ForEach{$l = 1,1+\eps,(1+\eps)^2,\ldots, k-s$}{
				$\tilde{S},\tilde{Q} \gets \textbf{ApproxLISBound}(z,b,\eps,l)$ \;
				for each $s'\in S'$ such that $s\leq s' \leq s+l$, let $\tilde{s}$ be the smallest element in $\tilde{S}$ that is larger than $s'-s$  and set $Q'[s'] = \min \{\tilde{Q}[\tilde{s}], Q'[s']\}$.\label{algo:approxlis_updateQS}
			}
		}
		$S\gets S'$, $Q\gets Q'$\;
	}

	\Return{$\max S$} \label{ApproxLIS:output}\;
	\caption{$\textbf{ApproxLIS}$ }
	\label{algo:approxlis}
\end{algorithm}

\begin{lemma}
	\label{lem:ApproxLIS}
	Given a sequence $x\in \Sigma^n$ and two parameters $b\leq \sqrt{n}$ and $\eps\in (0,1)$, $\textbf{ApproxLIS}(x,b,\eps)$ computes a $(1-3\log_b (n) \eps)$ approximation of $\lis(x)$ with $O(\frac{b\log^2 n}{\eps \log b})$ bits of space in $(O(\frac{b^2}{\eps^2}\log n))^{\log_b n-2} ( b^2\log n +\frac{b\log n}{\eps\log b})$ time.
\end{lemma}

\begin{proof}[Proof of Lemma~\ref{lem:ApproxLIS}]

$\textbf{ApproxLIS}$ is a recursive algorithm. We start from level one and every time $\textbf{ApproxLIS}$ or $\textbf{ApproxLISBound}$ is called, we enter the next level. Assume the input string at the first level has length $n$. Notice that except the last level, we always divide the string evenly into $b$ blocks. So the length of input string at the $i$-th level is bounded by $\lceil \frac{n}{b^{i-1}}\rceil$. The recursion stops when the input string has length no larger than $b^2$. Thus, the depth of recursion is at most $\log_b n-1$.  In the following, we denote the depth of recursion by $d$. 

To prove the correctness of our algorithm, we show the following claim.

\begin{claim}
	\label{claim:approxLIScorrectness}
	At the $i$-th level of recursion, let $x$ denote the input string at this level. $\textbf{ApproxLIS}(x, b, \eps)$ outputs a $(1-3(d-i)\eps)$ approximationg of $\lis(x)$. For any $l$, if there is an increasing subsequence of $x$ with length $l$ and ending with $\alpha \in \Sigma$, then $ApproxLISBound(x, b, \eps, l)$ outputs a set $S$ and a list $Q$, such that there is an element $s\in S$ with 
	$$(1-3(d-i)\eps) l\leq s \leq l$$
	and $Q[s]\leq \alpha$. 
\end{claim}

\begin{proof}[Proof of Claim~\ref{claim:approxLIScorrectness}]
	For simplicity, we abuse the notation a little by denoting the input string at $i$-th level $x$. The proof is by induction on $i$ from $d$ to 1.
	
	For the base case $i= d$, the input string $x$ has length at most $b^2$. We run the $(1-\eps)$ approximation algorithm from \cite{gopalan2007estimating}. The claim holds trivially by the correctness of their algorithm. 
	
	Assume the claim holds for $i+1$-th level for $1\leq i \leq d-1$. We now proof it also holds for $i$-th level.  Let $x$ be the input string at the $i$-th level. We start by showing the correctness of $\textbf{ApproxLIS}$.
	
	For our analysis, let $\tau$ be one of the longest increasing subsequence of $x$. $\tau$ can be divided into $b$ parts such that $\tau = \tau^1\circ \tau^2\circ \cdots \circ \tau^b$ and $\tau^i$ lies in $x^i$. We define the following variables. 
	
	$\alpha_i$ is the first symbol of $\tau^i$ (if $\tau^i$ is not empty). 
	
	$\beta_i$ is the last symbol of $\tau^i$ (if $\tau^i$ is not empty).
	
	$d_i = |\tau^i|$ is the length of $\tau_i$.
	
	$\gamma^i  = \tau^1\circ \tau^2\circ \cdots \circ \tau^i$ is the concatenation of the first $i$ blocks in $\tau$.
	
	$h_i = \sum_{j=1}^{i} d_j = |\gamma^i|$ is the length of $\gamma^i$.
	
	In the following, we let $P$ be the list we get after running \emph{PatienceSorting} with input $x$. $P'$ is the list ``interpolated" by $Q$ such that $P'[i] = Q[j]$ for the smallest $j\geq i$ that lies in $S$. If no such $j$ exist, set $P'[i] = \infty$. We denote the set $S$ and list $Q$ after processing the block $x^t$ (the $t$-th outer loop) by $S_t$ and $Q_t$ and the largest element in $S_t$ by $k_t$. Correspondingly, $P'_t$ is the list $P'$ after processing the $t$-th block $x^t$ and $P_t$ is the list after running \emph{PatienceSorting} with input $x^1\circ x^2 \circ \cdots \circ x^t$.
	
	
	Since $\tau$ is a longest increasing subsequence, without loss of generality, we can assume $P_t[h_t]  = \beta_t$ (if $\tau^t$ is not empty) for each $t$ from $1$ to $b$. This is because, if $P_t[h_t]  < \beta_t$, we can replace $\gamma^t$ with another increasing subsequence of  $x^1\circ x^2 \circ \cdots \circ x^t$ with length $h_t$ and ends with $P_t[h_t]$. On the other hand,  we must have $P_t[h_t]\leq \beta_t$ since $\gamma^t$ is an increasing subsequence of $x^1\circ x^2 \circ \cdots \circ x^t$ with length $h_t$. 
	
	We also assume that $P_t[h_t]  = P_{t+1}[h_t]$ if $\tau^t$ is an empty string ($h_t = h_{t+1}$). Since if not, we can replace $\gamma^{t+1}$ with another increasing substring with $\tau^t$ not empty. 
	
	We first show the following claim.
	
	\begin{claim}
		\label{claim:approxlis_correctness2}
		For each $t\in [b]$, we have
		$$ P'_t[(1-3(d-(i+1))\eps-\eps)h_t - 2t\frac{\eps}{b}k_t] \leq P_t[h_t].$$
	\end{claim}

	\begin{proof}[Proof of Claim~\ref{claim:approxlis_correctness2}]
		We prove this by induction on $t$. 
		
		For the base case $t=1$, if $d_1= 0$, then $h_1 = 0$. $P'_1[-2t\eps n^{-\frac{1}{d+1}}k_t]$ is not defined, we assume without loss of generality that $P'_1[\theta] = -\infty$ if $\theta\leq 0$. Since $P_1[0]$ and $P'_1[0]$ are both special symbol $-\infty$, the claim holds. 
		
		If $d_1 > 0$, we have $d_1 = h_1$. Let $l$ be the largest number such that $l = (1+\eps)^j$ for some integer $j$ and $l\leq d_1$. We have 
		$$\frac{1}{1+\eps}d_1\leq l \leq d_1.$$
		
		Let $\tilde{S},\tilde{Q}$ be the output of $\textbf{ApproxLISBound}(x^1,b, \eps,l)$. By our assumption on the correctness of Claim~\ref{claim:approxLIScorrectness} on $i+1$ recursive level, there exist an $\tilde{s}\in \tilde{S}$ such that 
		
		\begin{equation}
		\label{lis_a}
			\begin{split}
				\tilde{s}&\geq (1-3(d-(i+1))\eps)l \\
				& \geq \frac{1-3(d-(i+1)\eps)}{1+\eps} d_1 \\
				&\geq (1-3(d-(i+1))\eps-\eps) d_1
			\end{split}
		\end{equation}
		and 
		$$\tilde{Q}[\tilde{s}]\leq P_1[l]\leq P_1[h_1] = \beta_1.$$
		
		By the choice of $S_1$ (line~\ref{algo:approxlis_updateQS} of algorithm~\ref{algo:approxlis}), we know there is an $s\in S_1$ such that $\tilde{s}-\frac{\eps}{b}k_1\leq s \leq \tilde{s}$ and $Q_1[s]\leq  \tilde{Q}[\tilde{s}]$. 
		
		Combining \ref{lis_a}, we have 
		\begin{equation}
			\begin{split}
			P'_1[(1-3(d-(i+1))\eps-\eps)h_1-2\frac{\eps}{b}k_1]&\leq P'_1[\tilde{s}-2\frac{\eps}{b}k_1] \\
			&\leq P_1[s]\\
			&\leq \tilde{Q}[\tilde{s}] \\
			&\leq  P_1[h_1]
			\end{split}
		\end{equation} 
		
		This proved the base case of $t = 1$.
		
		Now we assume the claim holds for some fixed integer $t-1\leq b$, we show it also holds for $t$. If $\tau^t$ is an empty string, we have $h_t = h_{t-1} $ and $P_{t-1}[h_t] = P_{t}[h_t]$. Since $k_{t}\geq k_{t-1}$, we have 
		\begin{align*}
		P'_t[(1-3(d-(i+1))\eps-\eps)h_t - 2t\frac{\eps}{b}k_t] &\leq P'_t[(1-3(d-(i+1))\eps-\eps)h_t - 2t\frac{\eps}{b}k_{t-1}] \\ 
		&\leq   P'_{t-1}[(1-3(d-(i+1))\eps-\eps)h_t - 2t\frac{\eps}{b}k_{t-1}]\\
		&\leq P_{t-1}[h_t] = P_t[h_t]
		\end{align*}
		Thus, the claim holds for the case when $\tau^t$ is an empty string. 
		
		If $d_t>0$ ($\tau^t$ is not empty), by the assumption that 
		\begin{equation}
			P'_{t-1}[(1-3(d-(i+1))\eps-\eps)h_{t-1} - 2(t-1)\frac{\eps}{b}k_{t-1}]\leq P_{t-1}[h_{t_1}]		
		\end{equation}
		
		we know there is an $s_{a}\in S_{t-1}$ such that 
		\begin{equation}
		\label{lis_b}
			(1-3(d-(i+1))\eps-\eps)h_{t-1} - 2(t-1)\frac{\eps}{b}k_{t-1}-\frac{\eps}{b}k_{t-1}\leq  s_{a}\leq (1-3(d-(i+1))\eps-\eps)h_{t-1} - 2(t-1)\frac{\eps}{b}k_{t-1}
		\end{equation}
		and 
		\begin{equation}
			Q_{t-1}[s_{a}]  \leq P'_{t-1}[(1-3(d-(i+1))\eps-\eps)h_{t-1} - 2(t-1)\frac{\eps}{b}k_{t-1}]
		\end{equation} 
		Let $z$ be the subsequence of $x^t$ by only considering the elements larger than $Q[s_a]$. Similarly, we let $l$ be the largest number such that $l = (1+\eps)^j$ for some integer $j$ and $l\leq d_t$. That is 
		
		$$\frac{1}{1+\eps}d_t \leq l \leq d_t$$
		
		We run $\textbf{ApproxLISBound}(x^t,\eps,l)$ to get $\tilde{S}$ and $\tilde{Q}$. By our assumption on the correctness of $\textbf{ApproxLISBound}$ on the $(i+1)$-th level, there exist an $\tilde{s}\in \tilde{S}$ such that 

		\begin{equation}
			\label{lis_c}
			\begin{split}
			\tilde{s} & \geq (1-3(d-(i+1))\eps)l \\
			&\geq (1-3(d-(i+1))\eps)\frac{d_t}{1+\eps}\\
			&\geq (1-3(d-(i+1))\eps-\eps)d_t
			\end{split}
		\end{equation}
		
		and 
		
		$$\tilde{Q}[\tilde{s}]\leq P_t[s_a+l] \leq P_t[h_t] = \beta_t.$$ 
		
		Let $s_b$ be the largest element in $S_t$ such that $s_b\leq s_a+\tilde{s}$. We know $Q_t[s_b]\leq \tilde{Q}[\tilde{s}]\leq P_t[h_t]$ by the updating rule at line~\ref{algo:approxlis_updateQS} of algorithm~\ref{algo:approxlis}. By the choice of set $S_t$ and combining \ref{lis_b}, \ref{lis_c}, we have
		\begin{align*}
		s_b&\geq s_a+\tilde{s}-\frac{\eps}{b}k_t\\
		&\geq (1-3(d-(i+1))\eps-\eps)h_{t-1} - 2(t-1)\frac{\eps}{b}k_{t-1}-\frac{\eps}{b}k_{t-1} +  (1-3(d-(i+1))\eps-\eps)d_t -\frac{\eps}{b}k_t\\
		& \geq (1-\frac{2\eps}{3})h_{t}-2t\eps'n^{-\frac{1}{d+1}}k_t\\
		\end{align*}
		The last inequality is from the fact that $h_t = h_{t-1}+d_t$ and $k_t\geq k_{t-1}$. Since $P'_t[s_b]\leq P_t[h_t]$, we have shown that 
		$$ P'_t[(1-3(d-(i+1))\eps-\eps)h_t - 2t\frac{\eps}{b}k_t] \leq P_t[h_t]$$
		This finishes our proof of Claim~\ref{claim:approxlis_correctness2}.
	\end{proof}
	
	In Claim~\ref{claim:approxlis_correctness2}, when $t = b$, we have 
	$$ P'_b[(1-3(d-(i+1))\eps-\eps)h_b - 2t\eps k_b]  =  P'_b[(1-3(d-i)\eps)h_b ]\leq P_t[h_t]$$
	Thus, the output of $\textbf{ApproxLIS}$ at $i$-th level is at least a $(1-3(d-i)\eps)$ approximation of $\lis(x)$. 
	
	It remains to show the correctness of $\textbf{ApproxLISBound}$ at level $i$. 
	
	The analysis is essentially the same except now we replace the longest increasing subsequence $\tau$ with an increacing subsequence of length $l$ ending with the smallest possible symbol in $\Sigma$. 
	
	Assume $\tau$ ends with $\sigma\in \Sigma$, we know $\tau$ is the longest increasing subsequence ending with $\sigma$. Since otherwise, we can find some $\sigma'<\sigma$ such that there is an increasing subsequence of length $l$ ending with $\sigma'$. 
	
	Thus, we similarly define of $\alpha_i, \beta_i, d_i, \gamma_i, h_i$ for $i\in [b]$. We can assume $P_t[h_t] = \beta_t$ if $\tau^t$ is not empty and $P_{t-1}[h_{t-1}] = P_{t}[h_t]$ otherwise. The remaining analysis are mostly the same. We omit the details. 
	
\end{proof}

By Claim~\ref{claim:approxLIScorrectness}, at the first level, the output is a $(1-3d\eps)$ approximation of the length of LIS. Thus, $\textbf{ApproxLIS}(x, b, \eps)$ outputs a $(1-3\log_b (n) \eps)$ approximation of $\lis(x)$. 

We now turn to time and space complexity. Our algorithm $\textbf{ApproxLIS}$ is recursive and calls itself or $\textbf{ApproxLISBound}$, which has the same recursive structure. Each time $\textbf{ApproxLIS}$ and $\textbf{ApproxLISBound}$ is called at $(i-1)$-th recursive level, we enter the $i$-th level. When we enter the $i$-th level, if the input sequence has length larger than $b$, we need to call $\textbf{ApproxLIS}$ $O(b|S|)$ times and $\textbf{ApproxLISBound}$ $O(b|S|\log_{1+\eps}n)$ times. The recursion tree has degree $O(b|S|\log_{1+\eps}n) = O(\frac{b^2}{\eps^2}\log n)$

The order of computation is the same as doing a depth first search on the recursion tree. At each level of the recursion, we only need to remember the information in one node. For the leaf nodes, we run streaming algorithm from \cite{gopalan2007estimating} on a string of length at most $b^2$. It takes $O(\frac{b}{\eps}\log n)$ space. For the inner nodes, we need to maintain a set $S\subseteq [n]$ and a list $Q$, both has size $O(\frac{b}{\eps})$. Since each element in the set $S$ or list $Q$ takes $O(\log n)$ space. The space needed for one inner node is $O(\frac{b}{\eps}\log n)$. 

Since the depth of recursion is $d\leq \log_b n-1$, the total space for  $\textbf{ApproxLIS}$ is bounded by $O(d\frac{b}{\eps}\log n) =O(\frac{b\log^2 n}{\eps \log b})$.

For the time complexity, we first consider the time used within one level (excluding the time used for running $\textbf{ApproxLIS}$ and $\textbf{ApproxLISBound}$). We denote the time used within $i$-th level by $T_i$ and the total running time by $T$. We have $T = \sum_{i}^{d} T_i$.

For each node of the recursion tree, if the length of input string is at most $b^2$ (corresponding to the leaf nodes of the recursion tree), we run the streaming algorithm from \cite{gopalan2007estimating}. It takes time $O(b^2\log b)$. Since the recursion tree has degree $O(\frac{b^2}{\eps^2}\log n)$ and depth $d$, the number of nodes at the bottom level is bounded by $(O(\frac{b^2}{\eps^2}\log n))^{d-1}$. Since $d\leq \log_b n-1$, we have 

\begin{equation}
	\label{lis_d}
	T_d = (O(\frac{b^2}{\eps^2}\log n))^{\log_b n -2}b^2\log n.
\end{equation}

If the length of input string is more than $b^2$, the time is then dominated by the operations at line~\ref{algo:approxlis_updateQS} of algorithm~\ref{algo:approxlis}. Since the size of $S$ and $\tilde{S}$ are both at most $\frac{b}{\eps}$ and we try at most $\log_{1+\eps}n$ different $l$, it takes $O(b|S|^2\log_{1+\eps}n) = O(\frac{b^3}{\eps^3}\log n)$ time. Also, the number of nodes in $i$-th level is at most $O((\frac{b^2}{\eps^2}\log n)^{i-1})$. We know 

$$T_i = O((\frac{b^2}{\eps^2}\log n)^{i-1} \frac{b^3}{\eps^3}\log n).$$

Thus

\begin{equation}
	\label{lis_e}
	\begin{split}
	\sum_{i = 1}^{d-1} T_i & \leq (d-1)T_{d-1}\\
	& \leq  d(O(\frac{b^2}{\eps^2}\log n))^{d-2} \frac{b^3}{\eps^3}\log n ) \\
	& = (O(\frac{b^2}{\eps^2}\log n))^{\log_b n-2} \frac{b\log n}{\eps\log b})
	\end{split}
\end{equation}

Combining \ref{lis_d} and \ref{lis_e} and $d\leq \log_b n-1$, we know

\begin{equation}
	T = (O(\frac{b^2}{\eps^2}\log n))^{\log_b n-2} ( b^2\log n +\frac{b\log n}{\eps\log b})
\end{equation}

\end{proof}

\begin{theorem}
	\label{lem:lis_polylog}
	Given a string $x\in \Sigma^n$, there is a deterministic algorithm that computes a $1-O(\frac{1}{\log \log n})$ approximation of $\lis(x)$ with $O(\frac{\log^4 n}{\log \log n})$ bits of space in $O(n^{5+o(1)})$ time.
\end{theorem}

\begin{proof}[Proof of Theorem~\ref{lem:lis_polylog}]
	Let $b = \log n$ and $\eps = \frac{1}{\log n}$. Then Theorem~\ref{lem:lis_polylog} is a direct result of Lemma~\ref{lem:ApproxLIS}.
\end{proof}

\begin{theorem}
	\label{lem:lis_n^delta}
	Given a string $x\in \Sigma^n$, for any constant $\delta\in (0,\frac{1}{2})$ and $\eps\in (0,1)$, there is a deterministic algorithm that computes a $1-\eps$ approximation of $\lis(x)$ with $\tilde{O}_{\eps,\delta}(n^\delta)$ bits of space in $\tilde{O}_{\eps, \delta}(n^{2-2\delta})$ time.
\end{theorem}

\begin{proof}[Proof of Theorem~\ref{lem:lis_n^delta}]
	Let $b = n^\delta$ and pick $\eps'$ to be a constant sufficiently smaller than $\eps$. We run $\textbf{ApproxLIS}(x, b, \eps')$. Then Theorem~\ref{lem:lis_n^delta} is a direct result of Lemma~\ref{lem:ApproxLIS}.
\end{proof}

\subsection{Output the Sequence}

We can actually modify our algorithm to output the increasing subsequence we found (not only the length) at the cost of increased running time. We now show how it works. 

\begin{algorithm}
	\SetAlgoLined
	\DontPrintSemicolon
	\KwIn{A string $x $ ,parameters $b$ and $\eps$ }
	
	\If{$|x|\leq b$}{
		output the exact longest increasing subsequence with $O(b\log n)$ bits of space in $O(b\log n)$ time (see \cite{aldous1999longest} for example)
	}
	
	divide $x$ evenly into $b$ blocks such that $x = x^1\circ x^2\circ \cdots \circ x^b$ \;
	
	compute $ S_{b}$ and $Q_{b}$ by running $\textbf{ApproxLIS}(x, b , \eps )$  \Comment*[f]{$S_i$ and $Q_i$ are the set $S$ and list $Q$ after $i$-th outer loop of $\textbf{ApproxLIS}$ } \;

	set $B$ to be a list with $B[0] = -\infty$ and $B[b] = Q_b[s_b]$ where $ s_b= \max\{s\in S_b\}$\;
	
	\For{$i = b-1$ \textbf{to} $1$ \label{algo:lissequence_first_loop}}{
		
		release the space used for storing $S_{i+1}$, and $Q_{i+1}$\;
		
		compute $S_i$, $Q_i$ by running $\textbf{ApproxLIS}(x,\eps)$ \;
		
		\ForEach{$s\in S_i$ \textbf{such that } $s\leq s_{i+1}$ }{
			
			let $z$ be the subsequence of $x^i$ by only considering the elements larger than $Q[s]$\;
			
			\ForEach{$l = 1,1+\eps,(1+\eps)^2,\ldots, k-s$}{
				$\tilde{S},\tilde{Q} \gets \textbf{ApproxLISBound}(z,\eps,l)$ \;
				if there is an $\tilde{s}\in \tilde{S}$ such that $\tilde{s}+s\geq s_{i+1}$ and $B[i+1] = \tilde{Q}[\tilde{s}]$, we set $B[i] = Q_i[s]$, $s_i = s$ and \textbf{continue} \label{algo:lissequence_computeB} \;
			
			}	             
		}
	}
	\For{$i=1$ \textbf{to} $b$ \label{algo:lissequence_output_loop}}{
		let $z$ be the subsequence of $x^i$ ignoring every element larger than $B[i]$ or less or equal to $B[i-1]$\;
		$\textbf{LISSequence}(z,b, \eps)$\;
	}
	
	\caption{$\textbf{LISSequence}$}
	\label{algo:lissequence}
\end{algorithm}

We have the following result.

\begin{lemma}
	\label{lem:lissequence}
	Given a sequence $x\in \Sigma^n$ and two parameters $b\leq \sqrt{n}$ and $\eps\in (0,1)$, $\textbf{ApproxLIS}(x,b,\eps)$ outputs an increasing subsequence of $x$ with length at least $(1-3\log_b (n) \eps)\lis(x)$ with $O(\frac{b\log^2 n}{\eps \log b})$ bits of space in $(O(\frac{b^2}{\eps^2}\log n))^{d-1} ( b^3\log n +\frac{b^2\log n}{\eps\log b})$ time.
\end{lemma}

\begin{proof}[Proof of Lemma~\ref{lem:lissequence}]
	
	We start with the correctness of algorithm $\textbf{LISSequence}$. We prove this by showing that $\textbf{LISSequence}(x,b,\eps)$ outputs the longest increasing subsequence  detected by $\textbf{ApproxLIS}$. For example, after $i$-th loop, for any $s\in S_i$, we say our algorithm detected an increasing subsequence of length $s$ ending with character $Q_i[s]$. 
	
	Due to the limited space, we can not record all intermediate results $S_i$, $Q_i$, or the output of $\textbf{ApproxLISBound}$. 
	
	Let $S_b$ and $Q_b$ be the set and the list we get after running $\textbf{ApproxLIS}$ with inputs $x$, $b$, $\eps$. Let $s^b = \max S_b$. We know there must exists an increasing subseqeunce in $x$ with length at least $(1-3\log_b (n) \eps)\lis(x)$ ending with $Q[s^b]$. Although we have detected such a sequence but we only know its length and the last character. Our goal is to output this sequence. For convenience, we call this sequence $\rho$. 
	
	$\rho$ can be divided into $b$ blocks such that $\rho  = \rho^1\circ \rho^2\circ \cdots \circ \rho^b$ where $\rho^i$ lies in $x^i$.  Our goal is to recover a list $B$ such that the last character of $\rho^i$ is $B[i]$ (if $\rho^i$ is not empty). We set $B[b] = Q_b[s_b]$ since $\rho_b$ ending with $Q_b[s_b]$. Then we compute $B[i]$ from $i = b-1$ to $1$ by running $\textbf{ApproxLIS}$ and $\textbf{ApproxLISBound}$ multiple times. We can do the following.
	
	Assume we already know $s_{i+1}$ and $B[i+1]$, which means $\rho^1\circ \rho^2 \cdots \circ \rho^{i+1}$ has length $s_{i+1}$ ending with $B[i+1]$.

	By line~\ref{algo:approxlis_updateQS} of $\textbf{ApproxLIS}$, we know there must exist some $s'\in S_i$ and $l = (1+\eps)^j$ for some integer $j$, such that $s_{i+1} \leq s'+l$. We set $s_i $ to be such an $s'$. 
	
	By our choice of $s_i$, we know that there is some $l$ such that $\textbf{ApproxLISBound}(z,\eps,l)$ detects an increasing subsequence of length of length $s_{i+1}-s_i$ with first symbol larger than $Q_i[s_i]$ and last symbol at most $Q_{i+1}[s_{i+1}] = B[i+1]$. It suggests that $\rho^{i+1}$ has length $s_{i+1}-s_i$ and $\rho^i$ ends with $Q_i[s_i]$. Thus, we can set $B[i] = Q_i[s_i]$ and continue. 
	
	Once we have computed $B$, we know the first element of $\rho^i$ is larger than $B[i-1]$ and the last element at most $B[i]$. We can recursively use $\textbf{LISSequence}$ on $x^i$ with these restrictions to output $\rho^i$.

	For the space complexity, $\textbf{LISSequence}$ is also a recursive algorithm. It needs to call it self $b$ times. We start from the first level, every time we enter the next level, the length of input string is decreased by a factor of $b$. Thus, the recursion tree is of degree $b$ with depth at most $\log_b n$. The computation is in the same order as depth-first search on the recursion tree. We only need to remember the information of one node in each level. 
	
	For the leaf nodes, we do the exact compution with $O(b\log n)$ bits of space.
	 
	In each inner node, we maintain a list $B$ of size $b$. It takes $O(b\log n)$ space. We also run $\textbf{ApproxLIS}$ and $\textbf{ApproxLISBound}$ multiple times. By Lemma~\ref{lem:ApproxLIS}, this takes $O(\frac{b\log^2 n}{\eps \log b})$ bits of space. However, the space used for running $\textbf{ApproxLIS}$ and $\textbf{ApproxLISBound}$ can be reused. 
	
	Thus, the space complexity of algorithm $\textbf{LISSequence}$ is $O(\log_b n b\log n + \frac{b\log^2 n}{\eps \log b}) = O(\frac{b\log^2 n}{\eps \log b})$. This finishes the proof of Lemma~\ref{lem:lissequence}.

	For the time complexity, the running time can be divided into two parts: the time used for running $\textbf{ApproxLIS}$ and $\textbf{ApproxLISBound}$, and the time used by $\textbf{LISSequence}$ itself. 
	
	We start with the time used by $\textbf{LISSequence}$ (assuming $\textbf{ApproxLIS}$ and $\textbf{ApproxLISBound}$ are oracles and can get results in constant time).  $\textbf{LISSequence}$ is recursive. Its recursion tree has degree $b$ and depth at mast $\log_b n$. For each leaf node, it takes $O(b\log n)$ time and the number of leaf nodes is bounded by $b^{\log_b n-1}$. For each inner node, it computes list $B$. The time is dominated by the operations at line~\ref{algo:lissequence_computeB} of algorithm~\ref{algo:lissequence}. It takes $O(b|S||\tilde{S}|\log_{1+\eps}n) = O(\frac{b^3}{\eps^3}\log n)$ time. Also, the number of inner nodes the bounded by $\log_b nb^{\log_b n-2}$.
	
	Thus, the total running time used by $\textbf{LISSequence}$ itself is bounded by 
	
	\begin{equation}
		\label{lisseq_a}
		O(\log_b nb^{\log_b n-2}\frac{b^3}{\eps^3}\log n+b\log nb^{\log_b n-1} ) = O(\frac{b\log^n}{\eps^3}b^{\log_b n})
	\end{equation}

	Now we compute the time used for running $\textbf{ApproxLIS}$ and $\textbf{ApproxLISBound}$. Since $b$ and $\eps$ are fixed parameters. Let $f(m)$ denote time required for running $\textbf{ApproxLIS}$ or $\textbf{ApproxLISBound}$ once with input string length $m$. 
	
	Notice that when we compute $\textbf{ApproxLIS}$ or $\textbf{ApproxLISBound}$ with input string length $n$, we need to compute $\textbf{ApproxLISBound}$ with input string length $\frac{m}{b}$ $O(\frac{b^2}{\eps^2}\log n)$ times. Thus, we have
	
	\begin{equation}
		(\frac{b^2}{\eps^2}\log n)f(\frac{m}{b}) \leq f(m)
	\end{equation} 
	
	By Lemma~\ref{lem:ApproxLIS}, we know 
	
	\begin{equation}
	\label{lisseq_b}
		f(n) = O((\frac{b^2}{\eps^2}\log n)^{\log_b n -2} ( b^2\log n +\frac{b\log n}{\eps\log b}))
	\end{equation}
	
	At the first recursive level of $\textbf{LISSequence}$, we need to run $\textbf{ApproxLIS}$ $b$ times with input string length $n$ and $\textbf{ApproxLISBound}$ $O(b|S|\log_{1+\eps}n) = O(\frac{b^2}{\eps^2}\log n)$ times with input string length $\frac{n}{b}$.
	
	Thus, the time for running $\textbf{ApproxLIS}$ and $\textbf{ApproxLISBound}$ at first recursive level is bounded by $O(bf(n))$.
	
	At the $i$-th recursive level, we need to run $\textbf{ApproxLIS}$ $b^{i}$ times with input string length $\frac{n}{b^{i-1}}$ and $\textbf{ApproxLISBound}$ $O(b^{i-1} b|S|\log_{1+\eps}n) = O(\frac{b^{i+1}}{\eps^2}\log n)$ times with input string length $\frac{n}{b^i}$.
	
	The time for running $\textbf{ApproxLIS}$ and $\textbf{ApproxLISBound}$ at the $i$-th recursive level is bounded by $O(b^i f(\frac{n}{b^{i-1}})) = o(\frac{1}{\log_b n}bf(n))$.
	
	Since the depth of recursion is at most $\log_b n$, the total time used for running $\textbf{ApproxLIS}$ and $\textbf{ApproxLISBound}$ is bounded by $O(bf(n))$. Combining \ref{lisseq_a} and \ref{lisseq_b}, we know the total running time is bounded by 
	\begin{equation}
		O(bf(n)) = (O(\frac{b^2}{\eps^2}\log n))^{d-1} ( b^3\log n +\frac{b^2\log n}{\eps\log b})
	\end{equation}

\end{proof}

As a direct result of Lemma~\ref{lem:lissequence}, we have the following 2 lemmas. 

\begin{theorem}
	\label{lem:lissequence_polylog}
	Given a string $x \in \Sigma^n$, there is a deterministic algorithm that can output an increasing subsequence of length at least $(1-O(\frac{1}{\log \log n}))\lis(x)$ with $O(\frac{\log^4 n}{\log \log n})$ bits of space in $ O(n^{5+o(1)})$ time. 
\end{theorem}

\begin{proof}[Proof of Theorem~\ref{lem:lissequence_polylog}]
	Let $b = \log n$ and $\eps = \frac{1}{\log n}$. Then Theorem~\ref{lem:lissequence_polylog} is a direct result of Lemma~\ref{lem:lissequence}.
\end{proof}

\begin{theorem}
	\label{lem:lissequence_ndelta}
	Given a string $x \in \Sigma^n$, for any constants $\delta\in (0,\frac{1}{2})$ and $\eps\in (0,1)$, there is a deterministic algorithm that can output an increasing subsequence of length at least $(1-\eps)\lis(x)$ with $\tilde{O}_{\eps, \delta}(n^\delta)$ bits of space in $\tilde{O}_{\eps, \delta}(n^{2-\delta}) $ time. 
\end{theorem}

\begin{proof}[Proof of Theorem~\ref{lem:lissequence_ndelta}]
	Let $b = n^\delta$ and pick $\eps'$ to be a constant sufficiently smaller than $\eps$. We run $\textbf{LISSequence}(x, b, \eps')$. Then Theorem~\ref{lem:lissequence_ndelta} is a direct result of Lemma~\ref{lem:lissequence}.
\end{proof}

\section{Longest Common Subsequence}
\label{lcs}

In this section, we describe our algorithm for approximating $\lcs(x)$ with small space. Before introducing our algorithm, we introduce the following reduction from LCS to LIS.

\subsection{Reducing LCS to LIS}

Our space efficient algorithm for LCS is based on a reduction (algorithm~\ref{algo:ReduceLCStoLIS}) from LCS to LIS.

\begin{algorithm}
	\SetAlgoLined
	\DontPrintSemicolon
	\KwIn{Two strings $x\in \Sigma^n$ and $y \in \Sigma^m$.}
	\KwOut{An integer sequence $z\in [m]^*$ }
	initialize $z$ to be an empty string\;
	\For{$i = 1$ \textbf{to} $n$} {
		\For{$j = m$ \textbf{to} $1$}{
			if $x_i= y_j$, add $j$ to the end of $z$.
		}	
	}
	\Return{$z$}\;
	\caption{ReduceLCStoLIS}
	\label{algo:ReduceLCStoLIS}
\end{algorithm}

\begin{lemma}
	\label{lem:reduceLCStoLIS}
	Given two strings $x\in \Sigma^n$ and $y \in \Sigma^m$ as input to algorithm~\ref{algo:ReduceLCStoLIS}, let $z = \text{ReduceLCStoLIS}(x,y)\in [m]^*$ be the output, then the length of $z$ is $O(mn)$ and $\lis(z) = \lcs(x,y)$.
\end{lemma}

\begin{proof}[Proof of Lemma~\ref{lem:reduceLCStoLIS}]
	$z$ can be viewed as the concatenation of $n$ blocks such that $z = \hat{z}^1\circ \hat{z}^2\circ \cdots \circ \hat{z}^n$ ($\hat{z}^i$'s can be empty). For each $i$, the length of $\hat{z}^i$ is equal to the number of characters in $y$ that are equal to $x_i$. The elements of $\hat{z}^i$ are the indices of characters in $y$ that are equal to $x_i$ and the indices in $\hat{z}^i$ are sorted in descending order. Since the length of $\hat{z}^i$ for each $i$ is at most $m$, the length of $z$ is at most $mn$.
	
	Assuming $\lis(z) = l$, we show $\lcs(x,y)\geq l$. By the assumption, there exists a subsequence of $z$ with length $l$. We denote this subsequence by $t\in [m]^l$.  Let $t = t_1t_2\cdots t_l$. Since $\hat{z}^i$'s are strictly descending, eash element in $t$ is picked from a distinct block. We assume for each $i\in [l]$, $t_i$ is picked from the block $\hat{z}^{t'_i}$. Then by the algorithm, we know $x_{t'_i} = y_{t_i}$. For $1\leq i < j\leq l$, $t_i$ appears before $t_j$. The block  $\hat{z}^{t'_i}$ also appears before  $\hat{z}^{t'_j}$. We have $1\leq t'_1<t'_2<\cdots < t'_l\leq n$. Thus, $x_{t'_1}x_{t'_2}\cdots x_{t'_l}$ is a subsequence of $x$ with length $l$ and it is equal to $y_{t_1}y_{t_2}\cdots y_{t_l}$. Hence, $\lcs(x,y)$ is at least $l$.
	
	On the other direction, assuming $\lcs(x,y) = l$, we show $\lis(z)\geq l$. By the assumption, let $x' = x_{t'_1}x_{t'_2}\cdots x_{t'_l}$ be a subsequence of $x$ and $y' = y_{t_1}y_{t_2}\cdots y_{t_l}$ be a subsequence of $y$ such that $x' = y'$. Let $z' = \hat{z}^{t'_1}\circ \hat{z}^{t'_2}\circ \cdots \circ \hat{z}^{t'_l}$, which is a subsequence of $z$. For each $i\in [l]$, since $x_{t'_i}=y_{t_i}$, $t_i$ appears in the block $\hat{z}^{t'_i}$. By $1\leq t_1<t_2<\cdots < t_l\leq m$, we know $t = t_1t_2\cdots t_l$ is an increasing subsequence of $z'$ and thus also an increasing subsequence of $z$.
\end{proof}

\subsection{Space Efficient Algorithm for LCS}

Our goal is to compute the longest common subsequence between two strings $x$ and $y$ over alphabet $\Sigma$. We assume the input strings $x$ and $y$ both has length $n$ and the alphabet size $|\Sigma|$ is polynomial in $n$. We call our space efficient algorithm for LCS $\textbf{ApproxLCS}$ and give the pseudocode in algorithm~\ref{algo:approxlcs}.

The idea is to first reduce calculating $\lcs(x,y)$ to computing LIS with algorithm~\ref{algo:ReduceLCStoLIS}.  We do not use $\textbf{ApproxLIS}$ as a black box. Instead, we make slight modification to the approach to achieve better running time. We denote $z = \text{ReduceLCStoLIS}(x,y)$. Although storing $z= \text{ReduceLCStoLIS}(x,y)$ already takes $O(n^2\log n)$ bits of space, we will show later that this is not required for our algorithm. 

Similar to the case for LIS, we introduce a slightly modified version of $\textbf{ApproxLCS}$ called $\textbf{ApproxLCSBound}$. It takes an additional input $l$. The modification are same: first, at line~\ref{ApproxLCS:newk} of algorithm $\textbf{ApproxLCS}$, we always require $k$ to be at most $l$. That is, we let $k = \min\{l, \max\{k, s+d\}\}$. Second, instead of output $\max S$, we output the whole set $S$ and list $Q$ (at the bottom level, we output the list maintained by \emph{PatienceSorting}). We omit the pseudocode for $\textbf{ApproxLCSBound}$.

\begin{algorithm}
	\SetAlgoLined
	\DontPrintSemicolon
	\KwIn{Two strings $x, y \in \Sigma^*$, parameters $b$ and $\eps$}
	
	\If{$|x|\leq b $ }{
		let $z = \text{ReduceLCStoLIS}(x,y)$
		compute $\lis(z)$ exactly with $O(b\log n)$ space in $O(|x||y|) = O(bn\log n)$ time with \emph{PatienceSorting}. 
	}

	divide $x$ evenly into $b$ blocks such that $x = x^1\circ x^2\circ \cdots \circ x^b$\;
	
	initialize $S = \{0\}$ and $Q[0] = -\infty$ \;
	
	\For{$i = 1$ \textbf{to} $b$}{
		$k = 0$\;
		\ForEach{$s\in S$}{
			
			$d = \textbf{ApproxLCS} (x^i,y^*(Q[s]),b,\eps)$  \Comment*[f]{by  $y^*(Q[s])$, we mean the string we get by replacing the first $Q[s]$ elements of $y$ with a special symbol $*$ that does not appear in $x$}\;
			$k = \max\{k,s+d \}$\label{ApproxLCS:newk}\;
		}
		if $k\leq b/\eps$, let $S' = \{0,1,2,\ldots,k\}$, otherwise let $S' = \{0,\frac{\eps}{b}k,2\frac{\eps}{b}k,\ldots ,k\}$ \Comment*[f]{evenly pick $ \frac{b}{\eps}+1$ integers from $0$ to $k$ (including $0$ and $k$) } \label{algo:approxlcsd+1_newS}\;
		$Q'[s] \gets \infty$ for all $s\in S'$ except $Q'[0] = -\infty$\;
		\ForEach{$s\in S$}{
			\ForEach{$l = 1,1+\eps,(1+\eps)^2,\ldots, k-s$}{
				$\tilde{S},\tilde{Q} \gets \textbf{ApproxLCSBound}(x^i,y^*(Q[s]),b, \eps,l)$ \;
				for each $s'\in S'$ such that $s\leq s' \leq s+l$, let $\tilde{s}$ be the smallest element in $\tilde{S}$ that is larger than $s'-s$  and set $Q'[s'] = \min \{\tilde{Q}[\tilde{s}], Q'[s']\}$. \label{algo:approxlcs_updateQS}\;
			}
		}
		$S\gets S'$, $Q\gets Q'$\;
	}

	\Return{$\max\{s\in S\}$}\;
	\caption{$\textbf{ApproxLCS}$}
	\label{algo:approxlcs}
\end{algorithm}

\begin{lemma}
	\label{lem:approxlcs}
	Given two strings $x, y \in \Sigma^n$, parameters $b\leq \sqrt{n}$ and $\eps\in(0,1)$, $\textbf{ApproxLCS}(x,y,b,\eps)$ computes a $(1-3\log_b (n) \eps)$ approximation of $\lcs(x,y)$ with $O(\frac{b\log^2 n}{\eps \log b})$ bits of space in $(O(\frac{b^2\log n}{\eps^2}))^{\log_b n-1}bn\log n$ time.
\end{lemma}

\begin{proof}
	When the input string $x$ is of length at most $b$, let $z = \text{ReduceLCStoLIS}(x,y)$. Since $z$ is consists of at most $b$ parts, all in decreasing order. $\lis(z)\leq b$. We can compute $\lis(z)$ using \emph{PatienceSorting} with $O(b\log n)$ space and $O(bn\log n)$ time. We do not need to store $z$. This is because \emph{PatienceSorting} only need to scan $z$ from left to right once.  We can do this by scanning $y$ from right to left $b$ times. The total time is still $O(bn\log n)$.
	
	Otherwise, let $z = \text{ReduceLCStoLIS}(x,y)$. We use the same notation as in the proof of Lemma~\ref{lem:reduceLCStoLIS} such that $z\in [n]^{O(n^2)}$ can be viewed as the concatenation of $n$ blocks. That is, $z = \hat{z}^1\circ \hat{z}^2\circ \cdots \circ \hat{z}^n$, where $\hat{z}^i$ consists of indices of characters in $y$ that are equal to $x^i$, arranged in descending order. 
	
	The recursion stops when the first input string has length no larger than $b$. Also, every time we enter next level, the length of input string $x$ is decreased by a factor of $b$. The depth of recursion is at most $\log_b n$. 
	
	Our algorithm $\textbf{ApproxLCS}$ is essentially computing the length of LIS of $z$. However, unlike the algorithm $\textbf{ApproxLIS}$, instead of partition $z$ equally into $b$ blocks, we partition $z$ according to $x$. That is, we first evenly divide $x$ into $b$ blocks such that $x = x^1\circ x^2 \circ \cdots \circ x^b$. $z$ is then naturally divided into $b$ blocks $z = z^1\circ z^2 \circ \cdots \circ z^b$ (note that $z^i$ is not the same as $\hat{z}^i$), where $z^i = \text{ReduceLCStoLIS}(x^i, y)$. Thus, $\textbf{ApproxLCS}(x^i, y, b, \eps)$ computes a good approximation of $\lis(z^i)$.
	
	In our algorithm, we use the notation $y^*(Q[s])$ to denote the string we get by replacing the first $Q[s]$ characters of $y$ with a special symbol $*$ that does not appear in $x$. Thus, $\text{ReduceLCStoLIS}(x^i, y^*(Q[s]))$ is the subsequence of $\text{ReduceLCStoLIS}(x^i, y)$ with only elements larger than $Q[s]$. By running $\textbf{ApproxLCS}(x^i, y^*(Q[s]), b, \eps)$, we are computing a good approximation of the length of LIS of $z^i$ with first element larger than $Q[s]$.
	
	The proof of correctness of algorithm $\textbf{ApproxLCS}$ then follows directly from that of algorithm $\textbf{ApproxLIS}$. 
	
	Notice that it is not required to stored string $z$ at any level of the algorithm. We divide $z$ according to the corresponding position in $x$ and we only need to query $z$ at the last level. 
	
	We now turn to space and time complexity. The analysis is similar to that of algorithm $\textbf{ApproxLIS}$ except

	$\textbf{ApproxLCS}$ is a recursive algorithm. We start by analyse the recursion tree. Notice that except at the bottom level, $\textbf{ApproxLCS}$ needs to call itself $O(b|S|) = O(\frac{b^2}{\eps})$ times and $\textbf{ApproxLCSBound}$ $O(b|S|\log_{1+\eps}n) = O(\frac{b^2\log n}{\eps^2})$ (since $|S| $ is at most $\frac{b}{\eps}$) times. Thus, the degree of the recursion tree is $O(\frac{b^2\log n}{\eps^2})$. Also, as we have shown, the depth of recursion is at most $\log_b n$. 
	
	The computation of $\textbf{ApproxLCS}$  has the same order as doing depth-first search on the recursion tree. We only need to remember the information in one node at each level. 
	
	For the inner nodes of the recursion, $\textbf{ApproxLCS}$ maintains a set $S$ and a list $Q$, both of size $\frac{b}{\eps}$. Since each elements in $S$ and $Q$ takes $O(\log n)$ bits. The space needed for an inner node is $O(\frac{b}{\eps}\log n)$. For the leaf node,  we compute $\lcs(x,y)$ exactly with $O(b\log n)$ space. Thus, the total space required for $\textbf{ApproxLCS}$  is $O(d(\frac{b}{\eps}\log n))$ where $d$ is the depth of recursion. Since $d\leq \log_b n= \frac{\log n}{\log b}$, we know running $\textbf{ApproxLCS}$ takes $O(\frac{b\log^2 n}{\eps \log b})$.

	For the time complexity, we denote the used within $i$-th level by $T_i$ (excluding the time used for running itself or $\textbf{ApproxLCSBound}$) and the total running time by $ T$. We have $T = \sum_{i = 1}^d T_i$. 
	
	For the leaf nodes, we run exact algorithm with $O(bn\log n)$ time. Since the recursion tree has degree $O(\frac{b^2\log n}{\eps^2})$ and depth $\log_b n$, the number of nodes at $d$-th level (leaf nodes) is bounded by $(O(\frac{b^2\log n}{\eps^2}))^{\log_b n-1}$. We have
	\begin{equation}
	\label{lcs_a}
		T_d = (O(\frac{b^2\log n}{\eps^2}))^{\log_b n-1}bn\log n.
	\end{equation}
	
	For the inner nodes, the time is dominated by the operations at line~\ref{algo:approxlcs_updateQS} of algorithm~\ref{algo:approxlcs}. Since the size of $S$ and $\tilde{S}$ are both at most $\frac{b}{\eps}$ and we try at most $\log_{1+\eps}n$ different $l$, it takes $O(b|S|^2\log_{1+\eps}n) = O(\frac{b^3}{\eps^3}\log n)$ time. Also, the number of nodes at $i$-th level is bounded by $(O(\frac{b^2\log n}{\eps^2}))^{i-1}$. We have
	
	\begin{equation}
		T_i = (O(\frac{b^2\log n}{\eps^2}))^{i-1}\frac{b^3}{\eps^3}\log n.
	\end{equation}
	
	Thus
	
	\begin{equation}
	\label{lcs_b}
	\begin{split}
	\sum_{i = 1}^{d-1} T_i & \leq (d-1)T_{d-1}\\
	& \leq d(O(\frac{b^2\log n}{\eps^2}))^{d-2}\frac{b^3}{\eps^3}\log n\\
	& = (O(\frac{b^2\log n}{\eps^2}))^{d-1}\frac{b\log n}{\eps\log b}\\
	& = (O(\frac{b^2\log n}{\eps^2}))^{\log_b n-1}\frac{b\log n}{\eps\log b}
	\end{split}
	\end{equation}
	
	Compare \ref{lcs_a} and \ref{lcs_b}, we know that the total running time is dominated by $T_d$. We have
	
	$$T = (O(\frac{b^2\log n}{\eps^2}))^{\log_b n-1}bn\log n.$$

\end{proof}

\begin{theorem}
	\label{lem:lcs_polylog}
	Given two strings string $x, y\in \Sigma^n$, there is a deterministic algorithm that computes a $1-O(\frac{1}{\log \log n})$ approximation of $\lcs(x,y)$ with $O(\frac{\log^4 n}{\log \log n})$ bits of space in $O(n^{6+o(1)})$ time.
\end{theorem}

\begin{proof}[Proof of Theorem~\ref{lem:lcs_polylog}]
	Let $b = \log n$ and $\eps = \frac{1}{\log n}$. Then Theorem~\ref{lem:lcs_polylog} is a direct result of Lemma~\ref{lem:approxlcs}.
\end{proof}

\begin{theorem}
	\label{lem:lcs_n^delta}
	Given two strings $x,y\in \Sigma^n$, for any constant $\delta\in (0,\frac{1}{2})$ and $\eps\in (0,1)$, there is a deterministic algorithm that computes a $1-\eps$ approximation of $\lcs(x,y)$ with $\tilde{O}_{\eps,\delta}(n^\delta)$ bits of space in $\tilde{O}_{\eps, \delta}(n^{3-\delta})$ time.
\end{theorem}

\begin{proof}[Proof of Theorem~\ref{lem:lcs_n^delta}]
	Let $b = n^\delta$ and pick $\eps'$ to be a constant sufficiently smaller than $\eps$. We run $\textbf{ApproxLCS}(x, y, b, \eps')$. Then Theorem~\ref{lem:lcs_n^delta} is a direct result of Lemma~\ref{lem:approxlcs}.
\end{proof}

\subsection{Output the Subsequence}

We now show how to output the common subsequence we have detected with small space. The idea is similar to our approach on how to output longest increasing subsequence. 

Similarly, let $b$ be a parameter we will pick later. For the base case, we use the linear space alogrithm from \cite{hirschberg1975linear} that output a LCS  of $x$ and $y$ with $O(\min(n,m)\log n)$ space (we assume alphabet size is polynomial in $n$). Thus, one of the string has length no larger than $b$, we can output the longest common subsequence with $O(b\log n)$ space.  

We call our algorithm for outputing the sequence $\textbf{LCSSequence}$. The pseudocode can be found in algorithm~\ref{algo:lcssequence}.

\begin{algorithm}
	\SetAlgoLined
	\DontPrintSemicolon
	\KwIn{Two strings $x, y\in \Sigma^*$ and parameters $b$ and $\eps$ }
	
	\If{$|x|\leq b\textbf{ or } |y|\leq b$}{
		Output the longest common subsequence of $x$ and $y$ with $O(b\log n)$ bits of space and $O(bn )$ time.
	}
	
	divide $x$ evenly into $b$ blocks $x^1\circ x^2\circ \cdots \circ x^b$\;
	
	compute $ S_{b}$ and $Q_{b}$ by running $\textbf{ApproxLCS}(x,y,b, \eps)$  \Comment*[f]{$S_i$ and $Q_i$ are the set $S$ and list $Q$ after $i$-th outer loop of $\textbf{ApproxLCS}$ } \;
	set $B$ to be a list with $B[0] = -\infty$ and $B[b] = Q_b[s_b]$ where $ s_b= \max\{s\in S_b\}$\;
	\For{$i = b-1$ \textbf{to} $1$ \label{algo:lcssequence_first_loop}}{
		release the space used for storing $S_{i+1}$, and $Q_{i+1}$\;
		compute $S_i$, $Q_i$ by running $\textbf{ApproxLCS}^d(x,y,\eps)$ \;
		
		\ForEach{$s\in S_i$ \textbf{such that } $s\leq s_{i+1}$ }{
			let $z$ be the subsequence of $x^i$ by only considering the elements larger than $Q[s]$\;
			\ForEach{$l = 1,1+\eps,(1+\eps)^2,\ldots, k-s$}{
				$\tilde{S},\tilde{Q} \gets \textbf{ApproxLCSBound}(x^i,y^*(Q[s]),\eps,l)$\Comment*[f]{by  $y^*(Q[s])$, we mean the string we get by replacing the first $Q[s]$ elements of $y$ with a special symbol $*$ that does not appear in $x$}\;
				if there is an $\tilde{s}\in \tilde{S}$ such that $\tilde{s}+s\geq s_{i+1}$ and $B[i+1] = \tilde{Q}[\tilde{s}]$, we set $B[i] = Q_i[s]$, $s_i = s$ and \textbf{continue}  \;
			}	             
		}
	}
	\For{$i=1$ \textbf{to} $b$ \label{algo:lcssequence_output_loop}}{
		
		$\textbf{LCSSequence}(x^i,y_{[B[i-1]+1,B[i]]},b,  \eps)$\;
	}
	
	\caption{$\textbf{LCSSequence}$}
	\label{algo:lcssequence}
\end{algorithm}

\begin{lemma}
	\label{lem:lcssequence}
	
	Given two strings $x, y\in \Sigma^n$ and two parameters $b\leq \sqrt{n}$ and $\eps\in (0,1)$, $\textbf{LCSSequence}(x,y, b,\eps)$ outputs an increasing subsequence of $x$ with length at least $(1-3\log_b (n) \eps)\lcs(x,y)$ with $O(\frac{b\log^2 n}{\eps \log b})$ bits of space in $(O(\frac{b^2\log n}{\eps^2}))^{\log_b n-1}b^2n\log n$ time.
	
\end{lemma}

\begin{proof}
	[Proof of Lemma~\ref{lem:lcssequence}]
	Algorithm $\textbf{LCSSequence}$ is a modified version of $\textbf{LISSequence}$. One difference is that, when the input string $x$ has length at most $b$, we output the longest common subsequence of $x$ and $y$ with $O(b\log n)$ bits of space and $O(bn )$ time. This can be achieved by using the linear space algorithm from ~\cite{hirschberg1975linear}.
	
	Our algorithm for approximating LCS is base on the reduction from LCS to LIS. Let $z = \text{ReduceLCStoLIS}(x,y)$. Instead of output an increasing subsequence of $z$, we need to output the corresponding common subsequence of $x$ and $y$. Similarly to our analysis of algoirhm $\textbf{LISSequence}$. Let $\rho$ be the longest increasing subsequence we have detected in $z$ and divide $\rho$ into $b$ blocks such that $\rho^i$ lies in $z^i = \text{ReduceLCStoLIS}(x^i, y)$. The list $B$ here serves the same purpose as in the algorithm $\textbf{LISSequence}$. $B[i]$ is equal to the last element of $\rho^i$. Thus, for the next level of recursion, we only need to consider the LCS between $x^i$ and $y_{[B[i-1]+1,B[i]]}$. Here, $y_{[B[i-1]+1,B[i]]}$ is the substring of $y$ from position $B[i-1]+1$ to $B[i]$. 
	
	By the correctness of algorithm $\textbf{ApproxLCS}$, we are guaranteed to output a common subsequence of $x$ and $y$ with length at least $(1-3\log_b (n) \eps)\lcs(x,y)$.
	
	For the time and space complexity, the analysis is also same as that of $\textbf{LISSequence}$. The space is dominated by the space used for running $\textbf{ApproxLCS}$. The time complexity is $O(bf(n))$ where $f(n)$ is the running time of $\textbf{ApproxLCS}$ when the first input string has length $n$. 	
	
\end{proof}

\begin{theorem}
	\label{lem:lcssequence_polylog}
	Given two strings $x , y\in \Sigma^n$, there is a deterministic algorithm that can output an common subsequence of length at least $(1-O(\frac{1}{\log \log n}))\lcs(x,y)$ with $O(\frac{\log^4 n}{\log \log n})$ bits of space in $ O(n^{6+o(1)})$ time. 
\end{theorem}

\begin{proof}[Proof of Theorem~\ref{lem:lcssequence_polylog}]
	Let $b = \log n$ and $\eps = \frac{1}{\log n}$. Then Theorem~\ref{lem:lcssequence_polylog} is a direct result of Lemma~\ref{lem:lcssequence}.
\end{proof}

\begin{theorem}
	\label{lem:lcssequence_ndelta}
	Given two strings $x,y \in \Sigma^n$, for any constants $\delta\in (0,\frac{1}{2})$ and $\eps\in (0,1)$, there is a deterministic algorithm that can output an common subsequence of length at least $(1-\eps)\lcs(x,y)$ with $\tilde{O}_{\eps, \delta}(n^\delta)$ bits of space in $\tilde{O}_{\eps, \delta}(n^{3}) $ time. 
\end{theorem}

\begin{proof}[Proof of Theorem~\ref{lem:lcssequence_ndelta}]
	Let $b = n^\delta$ and pick $\eps'$ to be a constant sufficiently smaller than $\eps$. We run $\textbf{LCSSequence}(x, y, b, \eps')$. Then Theorem~\ref{lem:lcssequence_ndelta} is a direct result of Lemma~\ref{lem:lcssequence}.
\end{proof}

\section{Asymmetric streaming model}

\label{sec:asymmetric}

Asymmetric streaming model has been considered in \cite{saks2013space}. In this model, we have streaming access to one string and random access to the other string. We now show how our approaches can be used to get better algorithms in this model.

\subsection{Edit Distance}

\subsubsection{$(1+\eps)$ approximation with $\tilde{O}(\sqrt{n})$ space}

\begin{theorem}
	\label{lem:asymmetric_ed}
	Given two strings $x, y\in \Sigma^n$. Suppose we have streaming access to string $x$ and random access to string $y$. Then, there is a deterministic algorithm that, making one pass through $x$, outputs a $(1+\eps)$-approximation of $\ed(x,y)$ in $\tilde{O}(n^{2})$ time with $\tilde{O}(\sqrt{n})$ bits of space.
\end{theorem}

\begin{proof}[Proof of Theorem~\ref{lem:asymmetric_ed}]
	
	We take $b = \sqrt{n}$ and $\eps' $ to be a constant sufficiently smaller than $\eps$. 
	
	We now show that we can slightly modify the algorithm $\text{SpaceEfficientApproxED}$ (algorithm~\ref{algo:SpaceEfficientApproxED}) to make it work in the asymmetric model. 
	
	The idea is to run the for-loop starting from line~\ref{algo:SpaceEfficientApproxEDDeltaLoop} of algorithm~\ref{algo:SpaceEfficientApproxED} in parallel. This creates $O(\log_{1+\eps} n) = O(\frac{n}{\eps})$ parallel instances. Finally, we output the smallest edit distance find by these instances. This does not change the result of $\text{SpaceEfficientApproxED}(x,y, b,\eps')$. 
	
	We only need to show each instance takes $\tilde{O}(\sqrt{n})$ bits of space and reads $x$ from left to right once. 
	
	Notice that the computation in each instance is the same as running $\text{SpaceEfficientApproxED}$ except that we only try one $\Delta$ instead of all $\Delta = (1+\eps)^j$ for $j\in [\lceil \log_{1+\eps} n\rceil]$. Thus, it has the same space complexity as $\text{SpaceEfficientApproxED}$ and can be computed with $\tilde{O}(\sqrt{n})$ bits of space. Running them in parallel increase the aggregated space by a factor of $O(\log_{1+\eps} n)$
	
	The depth of recursion is two when $b = \sqrt{n}$. We only need to query $x$ in the second level and we query each block of $x$ one by one. That is, we only query block $x^{i+1}$ after we have finished the computation on input $x^i$. Thus, when we need to query block $x^i$, we can store the whole block with $O(\sqrt{n}\log n)$ bits of space.  After we have finished the computation on block $x^i$, we can release the space and scan the next $\sqrt{n}$ elements of $x$. This only adds another $O(\sqrt{n}\log n)$ bits to the aggregated space and we only need to scan $x$ from left to right once.

\end{proof}

\subsubsection{$\tilde{O}(n^\delta/\delta)$ space algorithm with polynomial time}

A recent work by Farhadi et. al. \cite{farhadi2020streaming} gives an asymmetric streaming algorithm that finds a $O(2^{\frac{1}{\delta}})$ approximation of edit distance using $\tilde{O}(\frac{n^\delta}{\delta})$ space at the expense of a running time exponential to the input size.

We now show that, combined with our space-efficient approximation algorithm, we can reduce the running time to a polynomial. 

\begin{theorem}
	\label{lem:asymmetric_ed_ndelta}
	Given two strings $x, y\in \Sigma^n$. Suppose we have streaming access to $x$ and random access to $y$. For any constants $\delta \in (0,\frac{1}{2})$, there is a deterministic algorithm that, making one pass through $x$, outputs a $O(2^{\frac{1}{\delta}})$-approximation of $\ed(x,y)$ in $\tilde{O}_{ \delta}(n^{4})$ time with $O(\frac{n^\delta}{\delta}\log n)$ bits of space.
\end{theorem}

\begin{proof}[Proof of Theorem~\ref{lem:asymmetric_ed_ndelta}]
	We start with explaining the ideas used in \cite{farhadi2020streaming} and then show why combined with our space efficient algorithms, the running time can be reduced to a polynomial.

	Given an online string $x$, our goal is to find an \emph{Approximate Closest Substring} of $x$ in $y$. For example, we say $y_{[l,r]}$ along with $d$ (an approximation of $\ed(x, y_{[l,r]})$), is an $\alpha$-approximation for the closest substring problem if for any substring $y_{[l^*, r^*]}$, we have
	\begin{equation}
		\ed(x, y_{[l,r]})\leq d \leq \alpha\ed(x, y_{[l^*,r^*]}).
	\end{equation}
	
	Let $\delta\in (0,\frac{1}{2})$ be a small constant and $b = n^\delta$. 
	
	If the online string $x$ has length at most $b = n^\delta$. We solve the closest substring problem exactly by finding a substring in $y$ that has the smallest edit distance to $x$. by computing the edit distance between $x$ and every substring of $y$ with length at most $2b$. Notice that compute the edit distance between two string with length $O(b)$ exactly takes $O(b^2)$ time and there are $O(nb)$ substrings of $y$ with length $2b$. This can be done with time  $O(b^3 n)$.
	
	If the online string $x$ has length larger than $b = n^\delta$. We first divide $x$ into $b$ blocks such that $x = x^1\circ x^2\circ \cdots \circ x^b$. Then, we find the approximate closest substring for each $x^i$ recursively. \cite{farhadi2020streaming} shows that, if for each block $x^i$, we are given an $\alpha$-approximation for the closest substring problem. That is, we know a substring $y_{[l_i,r_i]}$ of $y$ and a number $d_i$, such that for any substring $y_{[l^*, r^*]}$,
	
	\begin{equation}
	\label{asym_b}
		\ed(x^i, y_{[l_i,r_i]})\leq d_i \leq \alpha\ed(x, y_{[l^*,r^*]}).
	\end{equation}
	
	Then we can find a $(2\alpha+1)$-approximation for the closest substring of string $x$. That is, we can find a substring $y_{[l,r]}$ of $y$ and a number $d$, such that for any substring $y_{[l^*, r^*]}$,
	
	\begin{equation}
	\ed(x, y_{[l,r]})\leq d \leq (2\alpha+1)\ed(x, y_{[l^*,r^*]}).
	\end{equation}
	
	To do this, they try all possible $1\leq p_0\leq p_1\leq \ldots \leq p_b\leq n+1$ and set 
	
	\begin{equation}
		d = \min_{1\leq p_0\leq p_1\leq \ldots \leq p_b\leq n+1} (\sum_{i = 1}^{b} d_i + \ed(y_{[p_{i-1},p_i)}, y_{[l_i, r_i]}))
	\end{equation}
	
	and $l = p_0$ and $r = p_b-1$ with $p_0$, $p_b$ that achieve the minimum $d$. With triangle inequality, \cite{farhadi2020streaming} proved $y_{[l,r]}$ and $d$ is $(2\alpha+1 )$-approximation for the closest substring problem.
	
	The depth of recursion of the algorithm proposed by \cite{farhadi2020streaming} is at most $\frac{1}{\delta}$. Each recursive level adds a factor of two to the approximation of edit distance. Thus, this algorithm gives a $O(2^{\frac{1}{\delta}})$ approximation of the edit disance. Also notice that, the order of computation is the same as doing a depth first search on the recursion tree. To remember the partial results in each recursive level, we only need to store $l_i$, $r_i$ and $d_i$ for each $i\in [b]$. The space complexity is then bounded by $O(n^\delta/\delta\log n)$ bits.
	
	The super-polynomial running time comes from two parts: First, the algorithm tries all possible $1\leq p_0\leq p_1\leq \ldots \leq p_b\leq n+1$, there are $\binom{n}{n^\delta+1}$ choices for $p_0$ to $p_b$. Second, when computing $\ed(y_{[p_{i-1},p_i)}, y_{[l_i, r_i]})$, they use a $O(\log^2 n)$ space, quasi-polynomial time algorithm guaranteed by Savitch's theorem.
	
	We now show that with our space efficient algorithm for approximating edit distance, we can reduce the running time of above algorithm to polynomial, but with a slightly worse approximation factor.
	
	For convenience, let $\bar{y} = y_{[l_1,r_1]}\circ y_{[l_2, r_2]}\circ \cdots \circ y_{[l_b,r_b]}$. We do not need to store $\bar{y}$ with extra memory. Knowing $l_1, r_1, \ldots , l_b, r_b$, we can query each bit of $\bar{y}$ with an $O(\log n)$ time overhead. To achieve, we can precompute the accumulated length of the block $y_{[l_i,r_i]}$. This takes an additional $O(n^\delta\log n)$ bits of space. With this, to query a certain bit of $\bar{y}$, we can do a binary search to determine to query which bit in $y$ in $O(\log n)$ time. 
	
	Instead of trying all possible $1\leq p_0\leq p_1\leq \ldots \leq p_b\leq n+1$, we try all possible substrings $y_{[l',r']}$ with length at most $2|\bar{y}|$, i.e. $r'-l'\leq 2|x|$. And set $d$ to be 
	
	\begin{equation}
	\label{asym_a}
		d = \min_{r', l' \in [n], 1\leq r'-l'\leq 2|x|}\ed(\bar{y}, y_{[l',r']}) + \sum_{i = 1}^{b} d_i
	\end{equation}
	
	and $l = l'$ and $r = r'$ with $l'$, $r'$ that achieve the minimum value $d$. 
	
	Also, we do not compute $\ed(\bar{y}, y_{[l',r']}) $ exactly. Instead, we compute a $1+\eps$ approximation of $\ed(\bar{y}, y_{[l',r']}) $ with $O(n^\delta)$ space using our space efficient algorithm. Here, $\eps \in (0,1)$ is a small constant. 
	
	We show that the substring $y_{[l,r]}$ and $d$ is a $((2+\eps)\alpha + 1+\eps)$-approximation of closest substring of $x$. 
	
	To show this, let $y_{[\bar{l},\bar{r} ]}$ be the closest substring of $x$ in $y$. That is, for any substring $y_{[l^*, r^*]}$,
	
	\begin{equation}
	\ed(x, y_{[\bar{l},\bar{r}]})\leq \ed(x, y_{[l^*,r^*]}).
	\end{equation}
	
	We can assume $\bar{r}-\bar{l}\leq 2|x|$, since otherwise, $\ed(x, y_{[\bar{l},\bar{r}]})\geq |x|$ and we can replace $y_{[\bar{l},\bar{r} ]}$ with any substring of length $|x|$. 
	
	By the triangle inequality of edit distance, we know
	\begin{equation}
	\begin{split}
	\ed(\bar{y},y_{[\bar{l},\bar{r} ]} )&\leq  \ed(x, y_{[\bar{l},\bar{r}]})+ \ed(x,\bar{y}) \\
	&\leq  \ed(x, y_{[\bar{l},\bar{r}]}) + \sum_{i = 1}^{b} d_i. 
	\end{split}
	\end{equation}
	
	Here, the second inequality is guaranteed by the fact that 
	\begin{equation}
		\ed(x,\bar{y}) \leq \sum_{i = 1}^{b}\ed(x^i,y_{[l_i,r_i]}) \leq  \sum_{i = 1}^{b} d_i .
	\end{equation}
	
	Since $d$ in \ref{asym_a} takes minimum over $r', l' \in [n], 1\leq r'-l'\leq 2|x|$, the substring $y_{[\bar{l},\bar{r} ]}$ is also considered, we have 
	\begin{equation}
		d\leq (1+\eps)\ed(x,\bar{y}) + \sum_{i = 1}^{b} d_i\leq (1+\eps)\ed(x, y_{[\bar{l},\bar{r}]}) + (2+\eps)\sum_{i = 1}^{b} d_i.
	\end{equation}
	
	Notice that for any substring $y_{[l^*, r^*]}$, $y_{[l_i, r_i]}$ and $d_i$ satisfies equation~\ref{asym_b}. We have 
	\begin{equation}
	d\leq  ((2+\eps)\alpha + 1+\eps)\ed(x, y_{[\bar{l},\bar{r}]}) .
	\end{equation}
	
	Thus, $y_{[l,r]}$ and $d$ found by our method is a $((2+\eps)\alpha + 1+\eps)$-approximation of closest substring of $x$. Notice that this only requires an additional $O(n^\delta)$ space. The space complexity of thus unchanged asymptotically. 
	
	The structure of recursion remains the same. It will give us a $O((2+\eps)^{\frac{1}{\delta}})$ approximation of edit distance. Further more, if we take $\eps = \delta$, since $(2+\delta)^{\frac{1}{\delta}} = O(2^{\delta})$, the algorithm outputs a $O((2+\delta)^{\frac{1}{\delta}})$ approximation of edit distance. 
	
	For the time complexity, we consider the recursion tree, which has degree $b = n^\delta$ and depth $\log_b n = \frac{1}{\delta}$. 
	
	At the $\frac{1}{\delta}$-th level, we solve the closest substring problem by computing the edit distance between input at that level, $x$ ($|x|\leq b$), and every substring in $y$ with length at most $2b$. Notice that compute the edit distance between two string with length $O(b)$ exactly takes $O(b^2)$ time and there are $O(nb)$ substrings of $y$ with length $2b$. This can be done with time  $O(b^3 n)$. There are $b^{\frac{1}{\delta}-1}$ nodes at the $\frac{1}{\delta}$-th level. The total running time at that level is bounded by 
	
	\begin{equation}
		T_{\frac{1}{\delta}} = O(b^{\frac{1}{\delta}+2} n)  = O(n^{2+2\delta}) 
	\end{equation}
	
	For nodes in the $i$-th level of the recursion tree, the input string $x$ has length $\frac{n}{b^{i-1}}$. We need to consider every substring in $y$ with length at most $2|x|$ (there are $O(n|x|)$ such substring) and for each of these substring, we compute an $(1+\eps)$ approximation of edit distance between it and $x$. This takes $\tilde{O}_{\eps,\delta}(n|x|^3)$ time. Since there are $b^{i-1}$ nodes in the $i$-th recursive level, the total running time for $i$-th level is 
	\begin{equation}
		T_i = \tilde{O}_{\eps,\delta}(b^{i-1}n|x|^3) = \tilde{O}_{\eps,\delta}( \frac{n^4}{b^{2(i-1)}}) = \tilde{O}_{\eps,\delta}(n^{4-2\delta(i-1)})
	\end{equation}
	
	Thus, the running time is dominated by the computation at the first level. The time complexity is $\tilde{O}_{\eps,\delta}(n^4)$. Since we take $\eps = \delta$, the running time is $\tilde{O}_{\delta}(n^4)$.

\end{proof}

\subsection{Longest Common Subsequence}

\begin{theorem}
	\label{lem:asymmetric_lcs}
	Given two strings $x, y\in \Sigma^n$. Suppose we have streaming access to string $x$ and random access to string $y$. Then, there is a deterministic algorithm that, making one pass through $x$, outputs a $(1-\eps)$-approximation of $\lcs(x,y)$ in $\tilde{O}(n^{\frac{5}{2}})$ time with $O(\frac{\sqrt{n}}{\eps}\log n)$ bits of space.
\end{theorem}

\begin{proof}[Proof of Theorem~\ref{lem:asymmetric_lcs}]
	
	We take $b = \sqrt{n}$ and $\eps' $ to be a constant sufficiently smaller than $\eps$. 
	
	We will show $\textbf{ApproxLCS}(x, y, b, \eps')$ is such a streaming algorithm. When $b = \sqrt{n}$, the recursion has depth only 2. We need to query $x$ when running \emph{PatienceSorting} on $z^i$ at the second level of recursion. Also notice that to read $z^i$, we only needs to query the $i$-th block $x^i$. By Lemma~\ref{lem:approxlcs},  $\textbf{ApproxLCS}(x, y, b, \eps')$ can be computed in $\tilde{O}(n^{\frac{5}{2}})$ time with $O(\frac{\sqrt{n}}{\eps}\log n)$  bits of space.
	
	Notice that \emph{PatienceSorting} only need to read the input string from left to right once. Also, to read $z_i = \text{ReduceLCStoLIS}(x^i, y)$ from left to right once, we only need to scan $x^i$ from left to right once. Thus, $\textbf{ApproxLCS}$ makes only one pass through $x$ when $b = \sqrt{n}$. This proves the lemma.
\end{proof}

\section{Discussion and Open Problems}
\label{discussion}
In this paper we designed several space efficient approximation algorithms for three string measures that are widely used in practice: edit distance, longest common subsequence, and longest increasing sequence. All our algorithms are deterministic and can use space $n^{\delta}$ for any constant $\delta>0$, while achieving $1+\eps$ or $1-\eps$ approximation for any constant $\eps>0$ or even slightly sub constant $\eps$. The running time of our algorithms are essentially the same as, or only slightly larger than the standard algorithms which solve these problems exactly. With a larger polynomial running time, we can even achieve space complexity $\polylog(n)$. Our work leaves many interesting open problems, and we list them below.

\begin{enumerate}
\item Can we achieve better space complexity or better time complexity, or both? For example, is it possible to further reduce the space complexity to even logarithmic while still maintaining polynomial running time? Or  can we maintain poly-logarithmic space, but also achieve  quadratic or even sub-quadratic time complexity? What kind of approximations can we achieve in these cases? For example, can we keep the approximation factor to be $1+\eps$ or $1-\eps$, or a constant? We believe it requires new ideas to answer these questions. We remark that in this direction, a recent work \cite{chakraborty2019approximate} provides randomized algorithms which can give a constant factor approximation to ED in both slightly sub-linear space and slightly sub-quadratic time. It remains to see if one can do better or design a similar deterministic algorithm.

\item So far all our algorithms are deterministic. How does randomness help here? Can we design randomized algorithms that achieve $1+\eps$ or $1-\eps$ approximation, but with better space complexity?


\item Finally, is there a good reason for the lack of progress on computing edit distance and longest common subsequence \emph{exactly} using polynomial time and strongly sub linear space? In other words, it would be nice if one can provide justification like the SETH-hardness of computing edit distance and longest common subsequence exactly in truly sub-quadratic time. We note that a recent work of Yamakami \cite{Yamakami17} proposes a so called \emph{Linear Space Hypothesis}, which conjectures that some $\mathsf{NL}$-complete problems cannot be solved simultaneously in polynomial time and strongly sub linear space. Thus it would be nice to show reductions from these problems to edit distance and longest common subsequence.\ We note that here we need a reduction that simultaneously uses small space and polynomial time.
\end{enumerate}

\section*{Acknowledgement}
We thank Aviad Rubinstein and the authors of \cite{farhadi2020streaming} for pointing out an error in our first version of the paper.

\bibliographystyle{alpha}
\bibliography{references}

\appendix

\section{Proofs of results in section~\ref{massively_parallel}} \label{proofs}

\begin{proof}[Proof of Lemma~\ref{lem:approx_optimal_candidate}]
	Let $D_i = \ed(x^i,y_{[\alpha_i,\beta_i]})$. Since we assume the alignment is optimal and $[\alpha_i,\beta_i]$ are disjoint and span the entire length of $y$,  we know $\ed(x,y) = \sum_{i = 1}^{b}D_i$. 
	
	For each $i\in [b]$, if $\eps'|\alpha_i-\beta_i+1|\leq |x^i|\leq 1/\eps'|\alpha_i- \beta_i+1|$, by the defition of \emph{$(\eps',\Delta)$-approximately optimal condiate}, we know, 
	\begin{equation}
	\label{a}
	|\alpha_i-\alpha'_i|\leq \eps'\frac{\Delta}{b}
	\end{equation}
	and
	\begin{equation}
	\label{b}
	|\beta_i-\beta'_i| \leq \eps'\frac{\Delta}{b} + \eps' \ed(x^i, y_{[\alpha_i,\beta_i]})
	\end{equation}
	
	Also notice that we can transform  $y_{[\alpha'_i,\beta'_i]}$ to $y_{[\alpha_i,\beta_i]}$ with $|\alpha_i-\alpha'_i|+|\beta_i-\beta'_i| $ insertions and then transform $y_{[\alpha_i,\beta_i]}$ to $x^i$ with $\ed(x^i,y_{[\alpha_i,\beta_i]}) $ edit operations. We have
	\begin{equation}
	\label{c}
	\ed(x^i,y_{[\alpha'_i,\beta'_i]})\leq \ed(x^i,y_{[\alpha_i,\beta_i]}) + |\alpha_i-\alpha'_i|+|\beta_i-\beta'_i|
	\end{equation}
	
	Meanwhile, we can always transform $y_{[\alpha_i,\beta_i]}$ to $y_{[\alpha'_i,\beta'_i]}$ with $|\alpha_i-\alpha'_i|+|\beta_i-\beta'_i| $ deletions and then transform $y_{[\alpha'_i,\beta'_i]}$ to $x^i$ with $\ed(x^i,y_{[\alpha'_i,\beta'_i]})$. We have 
	\begin{equation}
	\label{f}
	D'_i\geq D_i.
	\end{equation}
	
	Combining \ref{a} \ref{b} \ref{c} and \ref{f}, we have
	
	\begin{equation}
	\label{d}
	D_i\leq D'_i\leq \ed(x^i,y_{[\alpha_i,\beta_i]}) + 2|\alpha_i-\alpha'_i|+2|\beta_i-\beta'_i| \leq (1+2\eps')D_i + 4 \eps'\frac{\Delta}{b}.
	\end{equation}
	
    For those $i$ such that $|x^i|>(1/\eps')|\alpha_i-\beta_i+1|$ or $ |x^i|<\eps'|\alpha_i-\beta_i+1|$, to transform $x^i $ to $y_{[\alpha_i,\beta_i]}$, we need to insert (or delete) $||\alpha_i-\beta_i+1|-|x^i||$ characters to make sure the length of $x^i$ equals to the length of $y_{[\alpha_i,\beta_i]}$. Thus, $D_i=\ed (x^i,y_{[\alpha_i,\beta_i]})$ is at least $ ||\alpha_i-\beta_i|-|l_i-r_i||$.  Since $D'_i = |\alpha_i-\beta_i|+|l_i-r_i|$, we have 
	\begin{equation}
	\label{e}
	\begin{aligned}
	D'_i\leq &\frac{1+\eps'}{1-\eps'}D_i &\\
	\leq & (1+3\eps')D_i & \text{Since we set } \eps' = \eps/10\leq 1/10
	\end{aligned}
	\end{equation}
	
	Also notice that we can turn $x^i$ into $y_{[\alpha_i,\beta_i]}$with $|l_i-r_i|$ deletions and $|\alpha_i-\beta_i|$ insertions, we know $D'_i\geq D_i$. It gives us 
	\begin{equation}
	\label{g}
	D_i\leq D'_i\leq (1+3\eps')D_i 
	\end{equation}
	
	Thus for each $i \in [b]$, by \ref{g} and \ref{d}, we have 
	\begin{equation}
	D_i\leq D'_i\leq (1+3\eps') D_i + 4\eps' \frac{\Delta}{N} .
	\end{equation}
	
	Since we assume $ \Delta\leq (1+\eps')\ed(x,y)$, we have $\eps'  \Delta\leq 1.1 \eps'\ed(x,y)$, this gives us 
	\begin{equation}
	\ed(x,y)\leq  \sum_{i = 1}^{b}D'_i  \leq  (1+3\eps') \ed(x,y)+ 4\eps' \Delta\leq (1+10\eps')\ed(x,y) = (1+\eps)\ed(x,y).
	\end{equation}

\end{proof}

\begin{proof}[Proof of Lemma~\ref{lem:CandidateSet}]
	Let $C^i_{\eps,\Delta}$ be the output of $\text{CandidateSet}(n,m,b,(l_i,r_i),\eps,\Delta)$. For the starting point $i'$, we only choose multiples of $\eps \frac{ \Delta}{b}$ from $[l_i-\Delta-\eps\frac{ \Delta}{b},l_i+\Delta+\eps\frac{ \Delta}{b}]$. At most $ O(\Delta/(\eps\frac{\Delta}{b}))  = O(b/\eps)$ starting points will be chosen. For each starting point, we consider $O(\log_{1+\eps} m) = O(\frac{\log m}{\eps})= O(\frac{\log n}{\eps})$ ending point since we assume $\eps m\leq n \leq \frac{1}{\eps} m$. Thus, the size of set $C^i_{\eps,\Delta}$ is at most $ O(\frac{b\log n}{\eps^2})$.
	
	We now show there is an element in $C^i_{\eps,\Delta} = \text{CandidateSet}(n,m,b,(l_i,r_i),\eps,\Delta)$ that is an \emph{$(\eps, \Delta)$-approximately optimal candidate} of $x^i$ if $\eps|\alpha_i-\beta_i+1|\leq |x^i|\leq 1/\eps|\alpha_i-\beta_i+1|$. 
	
	Since we assume $\Delta\geq \ed(x,y)$, we are guaranteed that $l_i-\Delta\leq \alpha_i\leq l_i+\Delta$. Thus, there is a multiple of $\lceil\eps \frac{ \Delta}{b}\rceil$, denoted by $\alpha'$, such that
	
	$$l_i-\Delta-\eps\frac{ \Delta}{b}\leq \alpha_i\leq \alpha' \leq \alpha_i+\eps \frac{ \Delta}{b}\leq l_i+\Delta+\eps\frac{ \Delta}{b},$$
	since we try every multiple of $\lceil\eps \frac{ \Delta}{b}\rceil$ between $l_i-\Delta-\eps\frac{ \Delta}{b}$ and $l_i+\Delta+\eps\frac{ \Delta}{b}$, one of them equals to $\alpha'$.
	
	For the ending point, we first consider the case when the length of  $y_{[\alpha_i,\beta_i]}$ is larger than the length of $x^i$, that is $\beta_i-\alpha_i+1\geq r_i-l_i+1$. We know $\ed(x^i,y_{[\alpha_i,\beta_i]})\geq \beta_i-\alpha_i+1 -|x^i|$. Let $j$ be the largest element in $\{0,1,\lceil 1+\eps\rceil, \lceil (1+\eps)^2\rceil,\cdots,\lceil (1+\eps)^{\log_{1+\eps}(m)}\rceil\}$ such that $\alpha'+|x^i|-1 +j\leq \beta_i$. We set $\beta' = \alpha'+|x^i|-1+j$. Since $j\geq (\beta_i-(\alpha'+|x^i|-1)/(1+\eps)$, we have
	\begin{equation}
	\begin{aligned}
	\beta' \geq& \alpha'+|x^i|-1+(\beta_i-(\alpha'+|x^i|-1))/(1+\eps)\\
	\geq & \frac{\beta_i}{1+\eps}+\frac{\eps}{1+\eps}(\alpha'+|x^i|-1)\\
	\geq & \beta_i - \frac{\eps}{1+\eps} (\beta_i-\alpha'+1-|x^i|)\\
	\geq &\beta_i - \eps \ed(x^i,y_{[\alpha_i,\beta_i]})
	\end{aligned}
	\end{equation}
	
	The last inequality is because $\ed(x^i,y_{[\alpha_i,\beta_i]})\geq \beta_i-\alpha_i+1-|x^i|\geq \beta_i-\alpha'+1-|x^i|$ and $\eps \geq \frac{\eps}{1+\eps} $. Thus, $(\alpha', \beta')\in C^i_{\eps,\Delta}$ is an \emph{$(\eps, \Delta)$-approximately optimal candidate} of $x^i$. 
	
	For the case when $\beta_i - \alpha_i + 1< |x^i|$. Similarly, we know $\ed(x^i,y_{[\alpha_i,\beta_i]})\geq  |x^i|-(\beta_i-\alpha_i+1)$.We pick $j$ to be the smallest element in $\{0,1,\lceil 1+\eps\rceil, \lceil (1+\eps)^2\rceil,\cdots,\lceil (1+\eps)^{\log_{1+\eps}(m)}\rceil\}$ such that $\alpha'+|x^i|-1-j\leq \beta_i$. We know $j\leq (1+\eps) (\alpha'+|x^i|-1-\beta_i)$. We set $\beta' = \alpha'+|x^i|-j$. Then
	\begin{equation}
	\begin{aligned}
	\beta' \geq& \alpha'+|x^i|-1-(1+\eps)(\alpha'+|x^i|-1-\beta_i)\\
	\geq & \beta_i-\eps(\alpha'-\beta_i+|x^i|-1)\\
	\geq & \beta_i - \eps(\alpha_i+\eps\frac{\Delta}{N}-\beta_i+|x^i|-1)\\
	\geq &\beta_i - \eps \ed(x^i,y_{[\alpha_i,\beta_i]})-\eps^2\frac{\Delta}{b}\\
	\geq &\beta_i - \eps \ed(x^i,y_{[\alpha_i,\beta_i]})-\eps\frac{\Delta}{b}
	\end{aligned}
	\end{equation}
	Thus, $(\alpha', \beta')\in C^i_{\eps,\Delta}$ is an \emph{$(\eps, \Delta)$-approximately optimal candidate} of $x^i$. 
	
\end{proof}

\begin{proof}[Proof of Lemma~\ref{lem:DPEditDistance}]
	We start by explaining the dynamic programming. Let $f$ be a function such that $f(i)\in C^i_{\eps',\Delta} \cup \{\emptyset\}$. We say an interval $x^i$ is matched if $f(i)\in C^i_{\eps',\Delta}$ and it is unmatched if $f(i) = \emptyset$. Let $S^f_1$ be the set of indices of matched blocks under function $f$ and $S^f_2 = [b]\setminus S^f_1$ be the set of indices of unmatched blocks. We let $f(i) = (\alpha^f_i,\beta^f_i)$ for each $i\in S^f_1$. We also require that,  for any $i,j\in S^f_1$ with $i<j$, $(\alpha^f_i,\beta^f_i)$ and $(\alpha^f_j,\beta^f_j)$ are disjoint and $\beta^f_i<\alpha^f_j$. Let $u_f$ be the number of unmatched characters under $f$ in $x$ and $y$. That is, $u_f$ equals to the number of indices in $[n]$ that is not in any matched block plus the number of indices in $[m]$ that is not in $f(i)$ for any $i\in S^f_1$.  Then we define the edit distance under match $f$ by \[\ed_f \coloneqq \sum_{i\in S^f_1}\ed(x^i, y_{[\alpha^f_i,\beta^f_i]})+u_f.\]
	
	Since we can always transform $x$ to $y$ by deleting (inserting) every unmatched characters in $x$ ($y$), and transforming each matched block $x^i$ into $y_{[\alpha^f_i,\beta^f_i]}$ with $\ed(x^i, y_{[\alpha^f_i,\beta^f_i]})$ edit operations. We know $\ed_f\geq \ed(x,y)$ 
	
	Let $F$ be the set of all matchings. Also, given $i\in [b]$ and $\alpha\in [m]$, we let $F^{i,\alpha}$ be the set of matching such that $f(i')$ is within $(1,\alpha)$ for all $i'\leq i$. Similarly, for each $f\in F^{i,\alpha}$, let $u^{i,\alpha}_f$ be the number of unmatched characters in $x_{[1,r_i]}$ and $y_{[1,\alpha]}$ under $f$. We can also define $\ed^{i,\alpha}_f = \sum_{i\in S^f_1}\ed(x^i, y_{[\alpha^f_i,\beta^f_i]})+u^{i,\alpha}_f$. For simplicity, let $C^i$ be the set of starting points of all intervals in $C^{i+1}_{\eps,\Delta}$. We now show that in algorithm~\ref{algo:DPEditDistance}, for each $i\in[b-1]$ and $\alpha\in C^{i+1}$, we have 
	
	\[A(i,\alpha-1) = \min_{f\in F^{i,\alpha-1}}\ed^{i,\alpha-1}_f.\] 
	
	We can proof this by induction on $i$. For the base case $i = 1$, we fix an $\alpha\in C^{2}$. For each $f\in F^{1,\alpha-1}$, if $f(1) = \emptyset$, then every character in $x^1$ and $y_{[1,\alpha-1]}$ are unmatched. In this case, $\ed^{1,\alpha-1}_f = |x^1|+\alpha -1 = A(0,\alpha'-1) + \alpha-\alpha' +|x^1|$ for every $\alpha'\in C^1$ such that $\alpha'\leq \alpha$. When $f(1) \neq \emptyset$, we assume $f(1) = (\alpha^f_1,\beta^f_1)$, then $\ed^{i,\alpha-1}_f = \alpha^f_1-1 + M(1,(\alpha^f_i,\beta^f_i))+ \alpha -\beta = A(0,\alpha^f)+ M(1,(\alpha^f_i,\beta^f_i))+ \alpha -\beta $. By the updating rule of $A(1,\alpha-1)$ at line~\ref{algo:DPEditDistance_updating}, we know $A(1,\alpha-1) = \min_{f\in F^{1,\alpha-1}}\ed^{1,\alpha-1}_f$ for every $\alpha\in C^{2}$.
	
	Now assume $A(t-1,\alpha-1) = \min_{f\in F^{t-1,\alpha-1}}\ed^{t-1,\alpha-1}_f$ for any $\alpha\in C^{t}$ for $1<t\leq b-1$. Fix an $\alpha_0\in C^{t+1}$, we show $A(t,\alpha_0-1) = \min_{f\in F^{t,\alpha_0-1}}\ed^{t,\alpha_0-1}_f$. For each matching $f$, if $f(t) = \emptyset$, we know 
	\[\ed^{t,\alpha_0-1}_f = \ed^{t-1,\alpha_0-1}_f + |x^t| \geq \min_{\alpha'\in C^{t},\alpha'\leq \alpha}A(t-1,\alpha'-1)+ \alpha_0-\alpha'\geq A(t,\alpha_0-1).\]
	When $f(t)\neq \emptyset$, we assume $f(t) = (\alpha^f_t,\beta^f_t)$. Then 
	\begin{align*}
	\ed^{t,\alpha_0-1}_f &= \ed^{t-1,\alpha^f_t-1}_f + M(t,(\alpha^f_t,\beta^f_t)) + \alpha_0-\beta^f_t-1\\
	& \geq A(t-1,\alpha^f_t-1)+M(t,(\alpha^f_t,\beta^f_t)) + \alpha_0-\beta^f_t-1\\
	&\geq A(t,\alpha_0-1)
	\end{align*}
	Meanwhile, $A(t,\alpha_0) \geq \min_{f\in F^{t,\alpha_0-1}}\ed^{t,\alpha_0-1}_f$ since $A(t,\alpha_0) = \ed^{t,\alpha_0-1}_f$ for some $f\in F^{t,\alpha_0-1}$ by the updating rule at line~\ref{algo:DPEditDistance_updating}. Thus, we have proved  $A(t,\alpha_0-1) = \min_{f\in F^{t,\alpha_0-1}}\ed^{t,\alpha_0-1}_f$.
	Now, assume we have computed $A(b-1,\alpha)$ for every $\alpha\in C^{b}$. Let $f_0$ be the optimal matching such that $\ed_{f_0}(x,y) = \min_{f\in F} \ed_f(x,y)$.  If $f_0(b) = \emptyset$, 
	\[\ed_{f_0}(x,y) = \min_{\alpha'\in C^{b}}A(b-1,\alpha'-1)+ |m-\alpha'| + |x^b|\]
	Otherwise, let $f_0(b) = (\alpha^{f_0}_b,\beta^{f_0}_b)$
	\[\ed_{f_0}(x,y)  = \min_{ (\alpha',\beta')\in C^b_{\eps,\Delta}}A(b-1,\alpha'-1)+M(b,(\alpha^{f_0}_b,\beta^{f_0}_b)) + m-\beta'\]
	By the optimality of $f_0$, we know the algorithm~\ref{algo:DPEditDistance} is $d = \ed_{f_0}(x,y)$. Now, if we fix an optimal alignment such that $x_{[l_i,r_i]}$ is matched  to block $y_{[\alpha_i,\beta_i]}$ and $[\alpha_i,\beta_i]$ are disjoint and span the entire length of $y$. Let $f_1$ be a matching such that, for each $i\in [b]$, if $\eps'|\alpha_i-\beta_i|\leq |l_i-r_i|\leq 1/\eps'|\alpha_i-\beta_i|$, $f(i)$ is an \emph{$(\eps', \Delta)$-approximately optimal candidate}. Otherwise, $f(i) = \emptyset$. By lemma~\ref{lem:approx_optimal_candidate} and lemma~\ref{lem:CandidateSet}, such a matching $f_1$ exists and $\ed_f\leq (1+\eps)\ed(x,y)$. Thus,
	\[\ed(x,y)\leq\ed_{f_0}\leq \ed_{f_1}\leq (1+\eps)\ed(x,y)\]
	This proves the correctness of algorithm~\ref{algo:DPEditDistance}. 
	
	Now we compute the tme complexity. By the proof of lemma~\ref{lem:CandidateSet}, $|C^i| = O(\frac{b}{\eps})$ for $i\in [b]$. The size of matrix $A$ is $O(\frac{b^2}{\eps})$ where the rows of $A$ are indexed by $i$ from $0$ to $b-1$ and for the $i$-th row, the columns are indexed by the elements in set $C^{i+1}$. We can divide the dynamic programming into roughly $b$ steps and for the $i$-th step, we compute the row indexed by $i$.  Assume we have already computed the row indexed by $i-1$ of $A$. We first set 
	$$A(i,\alpha-1) = \min_{\alpha'\in C^i,\alpha'\leq \alpha} A(i-1,\alpha'-1) + \alpha-\alpha'+ |x^i|$$
	for all $\alpha \in C^i$. This takes $O(|C^i||C^{i+1}|) = O(\frac{b^2}{\log^2 n})$ time. Then, we query each elements in the $i$-th row of $M$. Say we queried $M(i, (\alpha',\beta'))$, we update all $A(i,\alpha-1)$ such that $\alpha-1\geq \beta'$ by 
	$$A(i,\alpha-1) = \min\{A(i,\alpha-1),A(i-1,\alpha'-1)+M(i,(\alpha',\beta')) + \alpha-1 -\beta'\}.$$
	This takes $O(|C^i_{\eps', \Delta}||C^{i+1}|) = O(\frac{b^2}{\eps^3}\log n)$ time. So the $i$-th step takes $O(\frac{b^2}{\eps^3}\log n)$ time. SInce there are $b$ steps, the time complexity is bound by $O(\frac{b^3}{\eps^3}\log n)$.
	
	
	For the space complexity, notice when updating $A(i,\alpha)$, we only need the information of $A(i-1, \alpha'-1)$ for every $\alpha'\in C^{i}$. Thus, we can release the space used to store $A(i-2, \alpha''-1)$ for every $\alpha''\in C^{i-1}$. And for line~\ref{DPEditDistance:last}, we only need the information of $A(i-1, \alpha-1)$ for every $\alpha\in C^i$. From algorithm~\ref{algo:CandidateSet}, we know that for each $i$, we pick at most $ b/\eps$ points as the starting point of the candidate intervals. The size of  $C^i$ is at most $b/\eps$. Since each element in $A$ is a number at most $n$, it can be stored with $O(\log n)$ bits of space. Thus, the space required is $O(\frac{b}{\eps}\log n)$.
	
	If we replace $M(i,(\alpha,\beta))$ with a $(1+\gamma)$ approximation of $\ed(x^i,y_{[\alpha,\beta]})$. Each $M(i,(\alpha,\beta))$ will add at most an $\gamma\ed(x^i,y_{[\alpha,\beta]})$ additive error. The amount of error added is bounded by $\gamma \ed(x,y)$. Thus, $\text{DPEditDistance}(n,m,b, \eps',\Delta,M)$ outputs a $(1+\eps)(1+\gamma)$-approximation of $\ed(x,y)$. The time and space complexity is not affected.
\end{proof}

\end{document}